\documentclass[letterpaper, 10 pt, conference, dvipsnames]{ieeeconf}  

\IEEEoverridecommandlockouts                              

\usepackage{graphics} 
\usepackage{times} 
\usepackage{amsmath} 
\usepackage{amssymb}  
\makeatletter
\let\NAT@parse\undefined
\makeatother
\usepackage{booktabs}
\usepackage[numbers]{natbib}
\usepackage{float}
\usepackage[font={small}]{caption}
\usepackage[subrefformat=parens]{subcaption}
\usepackage{hyperref}
\usepackage{tikz,tikzscale}
\usepackage[super]{nth}
\usepackage[english]{babel}

\usepackage{amsthm}
\usepackage{adjustbox}
\usepackage{pifont}

\newcommand{\R}{\mathbb{R}}

\usetikzlibrary{positioning}
\usetikzlibrary{arrows.meta}

\graphicspath{ {./images/} }
\hypersetup{hidelinks}

\newtheorem{theorem}{Theorem}
\newtheorem{lemma}{Lemma}

\newtheorem{corollary}{Corollary}

\newtheorem{definition}{Definition}
\theoremstyle{definition}
\newtheorem{example}{Example}

\newcommand{\cmark}{\color{PineGreen}\ding{51}}%
\newcommand{\xmark}{\color{red}\ding{55}}%

\title{\LARGE \bf
On the properties of path additions for traffic routing
}

\author{Matteo Bettini$^{1,*}$ and Amanda Prorok$^{1}$
\thanks{$^*$Corresponding author}
\thanks{$^{1}$Matteo Bettini and Amanda Prorok are with the Prorok Lab and the Department of
Computer Science and Technology, University of Cambridge, Cambridge,
United Kingdom, e-mail: \{\tt\small{mb2389,asp45\}@cl.cam.ac.uk}}}%

\begin{document}

\maketitle
\thispagestyle{empty}
\pagestyle{empty}

\begin{abstract}
In this paper we investigate the impact of path additions to transport networks with optimised traffic routing. 
In particular, we study the behaviour of total travel time, and consider both self-interested routing paradigms, such as User Equilibrium (UE) routing, as well as cooperative paradigms, such as classic Multi-Commodity (MC) network flow and System Optimal (SO) routing.
We provide a formal framework for designing transport networks through iterative path additions, introducing the concepts of \textit{trip spanning tree} and \textit{trip path graph}. Using this formalisation, we prove multiple properties of the objective function for transport network design.
Since the underlying routing problem is NP-Hard, we investigate properties that provide guarantees in approximate algorithm design. Firstly, while Braess' paradox has shown that total travel time is \textit{not} monotonic non-increasing with respect to path additions under self-interested routing (UE), we prove that, instead, monotonicity holds for cooperative routing (MC and SO). This result has the important implication that cooperative agents make the best use of redundant infrastructure.
Secondly, we prove via a counterexample that the intuitive statement `adding a path to a transport network always grants greater or equal benefit to users than adding it to a superset of that network' is false. In other words we prove that, for all the routing formulations studied, total travel time is \textit{not} supermodular with respect to path additions. While this counter-intuitive result yields a hardness property for algorithm design, we provide particular instances where, instead, the property of supermodularity holds.  Our study on monotonicity and supermodularity of total travel time with respect to path additions provides formal proofs and scenarios that constitute important insights for transport network designers.

\end{abstract}

\begin{keywords}
Transport network design, Supermodularity, Traffic routing.
\end{keywords}

\section{INTRODUCTION}
Traffic routing is a core transportation problem. It is concerned with directing multiple agents, with aligned or conflicting interests, to their destinations. Relevant applications include: road traffic~\cite{patriksson2015traffic}, air traffic~\cite{bertsimas1998air}, crowd control~\cite{bellomo2008modelling}, and internet packets~\cite{ash1997dynamic}. 
With the advent of Connected Autonomous Vehicles (CAVs), we expect mobility paradigms to undergo a fundamental change, as we transition from (self-interested) human control to semi- (or fully) automated cooperative paradigms~\cite{bagloee2016autonomous}.
This will lead to a decrease in unpredictable routing behaviours caused by human decision making, paving the way for greater traffic efficiency~\cite{duarte2018impact, hyldmar2019fleet}. Despite of this shift from non-automated to semi- (or fully) automated traffic control, in most cases, the underlying infrastructure (e.g., transport network) remains the one that was originally designed for human control. 


%

While many research efforts focus on developing algorithms for traffic planning and control, the transport network is often regarded as a fixed constraint. On the other hand, as we can observe from historic trends, every change to mobility systems was followed by a change to the transport network infrastructure~\cite{wolf1996car}. We aim to anticipate this trend by investigating the properties of environment modifications under several (fixed) routing strategies. 

The problem of optimizing the transport network topology is known as the Network Design Problem (NDP)~\cite{leblanc1975algorithm}. This problem is generally NP-hard~\cite{magnanti1984network} and, thus, optimal solutions for large instances are computationally infeasible. Analysing the properties of the objective function, on the other hand, could provide us with additional tools for the design of approximation algorithms with known suboptimality bounds. In particular, for set functions that satisfy the property of diminishing returns, also known as sub/supermodularity, near-optimal approximations can be found~\cite{fujishige2005submodular, prorok2020robust}. It is thus fundamental to provide an analysis of the properties of objective functions for network design under different routing objectives and paradigms.

\begin{table}[t!]
\caption{Summary of our results on the properties of total travel time with respect to path additions in transport networks. Multi-Commodity (MC) and System Optimal (SO) are cooperative routing formulations while User Equilibrium (UE) is self-interested.}
\centering
\begin{adjustbox}{width=\columnwidth}
\begin{tabular}{r c c c p{.5\linewidth}}
\toprule
Total travel time w.r.t. path additions & MC & SO & UE \\
\midrule
Monotonic non-increasing & \cmark & \cmark & \xmark  \\
& Thm~\ref{theo:mc_non_increasing} & Thm~\ref{theo:so_non_increasing} & Braess~\cite{braess2005paradox}\\
\midrule
Supermodular (general case) & \xmark & \xmark & \xmark\\
& Thm~\ref{theo:mc_supermod} & Thm~\ref{theo:so_ue_supermod} & Thm~\ref{theo:so_ue_supermod}\\
Supermodular (special case) & \cmark & \cmark & \cmark\\
& Thm~\ref{theo:parallel_agnostic} & Thm~\ref{theo:parallel_routing} & Thm~\ref{theo:parallel_routing}\\
\bottomrule
\end{tabular}
\end{adjustbox}
\label{tab:results}
\end{table}
In this work, we study the theory of transport network design through iterative path additions to derive and categorize its formal properties. We model traffic as flow traversing a transport network represented as a graph. Under this formulation, we consider three different routing problems tasked with guiding each agent to its destination: 
\textit{i)} \textit{minimum cost multi-commodity network flow problem (MC)}~\cite{ahuja1988network}, where agents experience constant travel times, \textit{ii)} \textit{system optimal (SO) routing}~\cite{roughgarden2002bad}, where travel time is dependent on edge congestion and agents cooperate to minimize the systems' cost (total travel time), and \textit{iii)} \textit{user equilibrium (UE) routing}~\cite{roughgarden2002bad}, also known as Wardrop equilibrium~\cite{wardrop1952road}, where travel time is dependent on edge congestion and agents are routed to minimize their individual travel time. These routing paradigms are coordinated as they all require a centralised routing infrastructure, following three different optimisation objectives. We furtherly refer to MC and SO as \textit{cooperative} routing, as they optimise for the system's benefit, and to UE as \textit{self-interested} routing, as it optimises for individual agents. We are interested in analysing the properties of total travel time for these routing problems under path additions to the transport network. We thus introduce the concepts of \textit{trip spanning tree} and \textit{trip path graph} which allow us to formalise what we consider an initial transport network and the space of its possible path additions. This model is then leveraged to prove our properties of interest.

Firstly, while Braess~\cite{braess2005paradox} has shown that total travel time is not monotonic non-increasing with respect to path additions under self-interested UE routing, we prove in our formulation that this does \textit{not} occur for cooperative routing (MC and SO). This result has the important implication that cooperative agents make the best use of redundant paths in the network. 

Secondly, we prove that the intuitive statement `adding a path to a transport network always grants greater or equal benefit to users than adding it to a superset of that network' is false. In other words, we prove that total travel time is not supermodular with respect to path additions. This is valid both when agents are routed according to their own interest (UE) and when they are cooperating for the system's benefit (MC and SO). We provide a counterexample to support our proofs. While this counter-intuitive result states that the supermodularity of path additions cannot be leveraged in approximation algorithms for general network optimisation, we provide particular instances where, instead, the property of supermodularity holds. In this scenario only parallel path additions are allowed, making it suitable to represent network design settings where physical additions can overlap but virtual network paths are kept parallel (e.g., by creating dedicated lanes).

 Our results are summarised in Table~\ref{tab:results}. This work contains the following contributions:
\begin{itemize}
    \item We present the problem of network design via iterative path additions and we introduce a graph structuring framework for network design which leverages the new concepts of \textit{trip spanning tree} and \textit{trip path graph}. Braess' example~\cite{braess2005paradox} fits in this framework.
    \item We prove that, for cooperative agents, the total travel time is monotonic non-increasing with respect to path additions.
    \item We prove with a counterexample that, in the general case, the total travel time is \textit{not} supermodular with respect to path additions for both cooperative and self-interested agents.
    \item We introduce and prove a special case where supermodularity holds (i.e., restricted supermodularity) and can thus be leveraged for designing approximation algorithms for specific network design instances.
\end{itemize}

The paper is structured in the following way: Section~\ref{sec:related_work} goes through the relevant related work, Section~\ref{sec:preliminaries} discusses some preliminaries, in Section~\ref{sec:problem_formulation} our formulation is introduced, Section~\ref{sec:routing} goes through the routing paradigms used, Section~\ref{sec:path_additions} describes our formalization of the initial transport network and its extension via path additions, Section \ref{sec:properties} investigates the properties of total travel time with respect to path additions, discussing monotonicity and non-supermodularity, to then focus on particular supermodular instances, and, finally, Section~\ref{sec:conclusion} concludes the paper.

\section{RELATED WORK}
\label{sec:related_work}
The problem of designing transport networks for traffic routing has been studied for a long time and goes by the name of Network Design Problem (NDP). It is also known as the Road Network Design Problem (RNDP)~\cite{Yang1998} and Transport Network Design Problem (TNDP)~\cite{chen2011transport}. Extensive discussion on the topic can be found in~\cite{farahani2013review, magnanti1984network, Yang1998}. The problem is usually formulated as a bilevel optimisation task. At the higher level, the network designer can make decisions that impact the network infrastructure. These decisions can be discrete (e.g., adding and removing edges) or continuous (e.g., extending the edge capacities). The objective of the network designer is usually to minimise the total travel time and, thus, the network congestion. Decisions involve physical network modifications which can be subject to a construction budget. At the lower level, agents see the network modifications and distribute in the resulting network according to their routing objective. Note that the designer cannot control directly the agents' behaviour, but only influence it through environment optimisation. The designer's goal is to minimise the total system travel time while respecting the construction budget. In other words, make the best use of network modifications in order to improve the routing performance according to the routing objective of choice. The NDP can also be viewed in a game theoretic framework, where the bilevel task is a Stackelberg leader-follower game~\cite{simaan1973stackelberg} with agents in the lower level playing a cooperative (MC, SO) or non-cooperative (UE) game depending on the routing objective.

The nature of the decision variables at the higher level divides the NDP into the Continuous Network Design Problem (CNDP), which considers continuous network modifications, and the Discrete Network Design Problem (DNDP), which considers discrete network modifications. Typical decision variables for the CNDP are: road capacities~\cite{abdulaal1979continuous, marcotte1986network, chiou2005bilevel}, traffic light schedules~\cite{marcotte1983network, ziyou2002reserve}, and road tolls~\cite{cole2006much, yang1997sensitivity}. Typical decision strategies for the DNDP are: new road construction~\cite{leblanc1975algorithm, gao2005solution, drezner2003network}, lane allocation~\cite{zhang2007two, drezner2002using}, and road capacity expansion~\cite{poorzahedy1982approximate, poorzahedy2005application}. In this work, we introduce a novel formulation for the higher level of the DNDP, where we consider adding paths to a transport network. Therefore, the space of possible path additions constitutes our decision variables. The addition of a path can lead to the construction of zero or multiple new roads.

At the lower optimisation level, most existing works consider agents routed according to some form of User Equilibrium (UE) routing: a formulation that characterizes a non-cooperative game in which every agent is self-interested and chooses the fastest path for themself according to the current congestion, leading to a Nash equilibrium. Popular variations of UE are: stochastic UE~\cite{daganzo1977stochastic} which considers uncertainty in the travel time and dynamic UE~\cite{peeta2001foundations} which considers dynamic user demands. The choice of UE for routing in the NDP literature is motivated by its fairness (agents travelling from the same origin to the same destination experience the same travel time) and by the intrinsic self-interested nature of humans. On the other hand, we argue that, with the introduction of CAVs, we will see a shift towards cooperative routing, as human self-interest will be replaced by autonomous cooperation. One example of this phenomenon is the recent introduction of autonomous mobility-on-demand vehicles~\cite{spieser2014toward}, which, by cooperating through a central routing authority, are able to reduce traffic congestion. In the literature, few works use System Optimal (SO) routing for the lower level~\cite{ben1988general, dantzig1979formulating}. Motivated by these considerations, in this work we provide results for both cooperative (MC, SO) and self-interested (UE) routing paradigms.

Our problem is an instance of the DNDP. The DNDP is generally NP-Hard~\cite{magnanti1984network}. Its hardness is caused by a series of factors. Firstly, even the most simple bilevel problem with linear upper and lower levels is strongly NP-Hard~\cite{marcotte2005bilevel, ben1988general}. Secondly, even if both levels are convex, the convexity of the bilevel problem cannot be guaranteed~\cite{luo1996mathematical}. Lastly, the discrete formulation of our decision variables highlights the combinatorial nature of the problem. Therefore, it is computationally intractable to solve the problem exactly for large instances. A branch-and-bound algorithm and a SO relaxation for the lower bounds is proposed in~\cite{leblanc1975algorithm}. Later, another branch-and-bound solution for agents routed using stochastic UE is introduced by \cite{chen1991network}. However, these exact methods are computationally infeasible for large networks and, thus, some alternatives have been explored. Branch and backtrack heuristics have been used in~\cite{poorzahedy1982approximate}, while metaheuristic hybrids based on ant colony optimisation are proposed in~\cite{poorzahedy2005application, poorzahedy2007hybrid}. The genetic algorithm represents one effective metaheuristic solution~\cite{yin2000genetic}, which has also been applied to the CNDP~\cite{xu2009study, mathew2009capacity}.  

Although metaheuristics achieve supposedly good solutions, they do not provide suboptimality guarantees. This is a main obstacle for the adoption of such network design solutions by companies and traffic regulators~\cite{bagloee2017hybrid}: the absence of guarantees lowers the marketability of heuristics and their unknown performance undermines the users' trust. Furthermore, without an available optimal solution, the goodness of metaheuristics cannot even be measured a posteriori. On the other hand, by leveraging certain properties of the objective function, it is possible to design approximation algorithms with suboptimality bounds. Such algorithms present all the characteristics required by companies and traffic regulators: they are computationally efficient while providing goodness guarantees. In particular, for set functions that satisfy the property of diminishing returns,
also known as sub/supermodularity, near-optimal approximations can be found~\cite{fujishige2005submodular}. There exist many works that exploit sub/supermodularity in particular instances of network design. An efficient topology optimisation algorithm for information diffusion is designed in~\cite{khalil2014scalable}. In \cite{mehr2018submodular}, submodularity is leveraged to achieve near-optimal sensor placement in traffic networks. Submodular network design to achieve desired algebraic rigidity properties is explored in~\cite{shames2015rigid}. Restricted supermodularity for robust network design is proved in~\cite{gupta2019budget}. While these works are able to exploit sub/supermodularity for network design, they all leverage particular problem formulations designed for their specific instance. Instead, in this work, we provide results on the properties of objective functions for network design both in general and special cases under multiple routing formulations in order to provide a comprehensive analysis.

\section{PRELIMINARIES}
We now introduce some preliminary concepts in transportation research and set functions. In Section~\ref{sec:braess}, we discuss the Braess' paradox~\cite{braess2005paradox}, which proves via counterexample that travel time of self-interested agents is not monotonic non-increasing with respect to path additions. This work relates to our study as we investigate the same property for cooperative agents. In Section~\ref{sec:supermod_backgroud}, we introduce supermodularity, a fundamental property of set functions leveraged in approximation algorithm design~\cite{fujishige2005submodular}. This property will be discussed in the context of transport network path additions in Section~\ref{sec:properties}.
\label{sec:preliminaries}
\subsection{Braess' paradox}
\label{sec:braess}
One important result in traffic research is a routing paradox discovered by Braess~\cite{braess2005paradox} in 1968, which has since become known as the Braess' paradox. Braess' paradox states that, when users are routed according to UE, the total system travel time can increase after a new path is added to the network. In other words, it presents a proof that the total travel time of users routed in a network according to UE is not monotonic non-increasing with respect to path additions. We can see how this phenomenon is not really a paradox as, when users are behaving according to UE, they behave according to self-interest and thus probably do not make the best use of the transportation infrastructure available. Braess illustrates this through a simple yet somewhat pathological example. However, Braess' paradox has been shown to be commonplace in road networks~\cite{steinberg1983prevalence, pas1997braess} and in large random networks~\cite{valiant2010braess}.

\begin{figure}[ht!]
\centering

  \begin{tikzpicture}[node distance={40mm}, thick, 
    node/.style = {draw, circle, minimum size=5mm},
    edge_label/.style = {midway, above, sloped, color=black},
    edge/.style = {->, color=black}
    ] 
    
    \node[node] (1) {$s$}; 
    \node[node] (2) [above right=1.2cm and 2cm of 1] {$v$};
    \node[node] (3) [below right=1.2cm and 2cm of 1] {$w$};
    \node[node] (4) [above right=1.2cm and 2cm of 3] {$t$};
    
    \begin{scope}[>={Straight Barb[black]}]
    \draw[edge] (1) to  node[midway, above, color=black, sloped] {$10x$} (2);
    \draw[edge] (1) to  node[midway, below, sloped, color=black] {$50+x$} (3);
    \draw[edge] (2) to  node[midway, above, color=black, sloped] {$50+x$} (4);
    \draw[edge] (3) to  node[midway, below, sloped, color=black] {$10x$} (4);
    \draw[<-] (3) to  node[midway, above, sloped,  color=black] {$10 + x$} (2);
    \end{scope}

\end{tikzpicture}

\caption{The Braess' paradox network. Labels on the edges represent the flow-dependant travel time functions. This network shows that, when a flow of 6 must be routed from $s$ to $t$ according to UE, the total travel time increases after the addition of edge $(v,w)$.}
\label{fig:braess}
\end{figure}
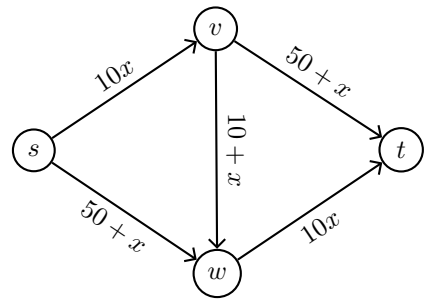

In Braess' example, shown in Fig.~\ref{fig:braess}, we consider a flow demand of 6 going from node $s$ to node $t$. The road travel time cost functions are reported on the respective edges of the graph. In the case where $(v,w)$ is not present in the network, the solutions to SO and UE routing coincide: all the users spread equally on the available paths, leading to a flow of 3 on the top edges and a flow of 3 on the bottom edges. Every user experiences the same travel time of $50 + 3 + 10 * 3 = 83$. This solution is both optimal for the users and for the system. On the other hand, if the edge $(v,w)$ is added to the network, the solution to the SO does not change, while self-interested users behaving according to UE see an opportunity to decrease their travel time by using the newly created path $(s,v) \rightarrow (v,w) \rightarrow (w,t)$. Thus, when the equilibrium is reached, we observe a flow of 2 in the previously available paths and a flow of 2 in the newly created one. Note that the flow on an edge is the sum of the flows on the paths passing through it, thus edge $(s,v)$ will have a flow of 4. In this UE scenario each user experiences the same travel time of $92$. The increase in the total travel time of users behaving accordingly to UE after the addition of edge $(v,w)$ is:
$92 * 6 - 83 * 6 = 552 - 498 = 54$.

Braess' paradox informs us that extreme cation must be used when structuring transport networks, as it presents a counterexample to prove that total travel time is not monotonic with respect to path additions under UE routing. Our work extends this result by proving that total travel time is not supermodular with respect to path additions under either MC, SO or UE routing.
\subsection{Supermodularity}
\label{sec:supermod_backgroud}
A fundamental component underlying this work is the property of supermodularity of set functions~\cite{topkis2011supermodularity}. Before introducing such property, let us first define non-increasing monotonicity.

\begin{definition}[Monotonic non-increasing set function]
Let $\Lambda$ be a set function defined as $\Lambda: 2^\mathcal{\mathcal{G}} \mapsto \R$, where $2^\mathcal{G}$ is the power set of the finite set $\mathcal{G}$. $\Lambda$ is a monotonic non-increasing function of $\mathcal{G}$ if and only if for every $\mathcal{A}$ and $\mathcal{B}$ such that $\mathcal{A} \subseteq \mathcal{B} \subseteq \mathcal{G}$, we have that $\Lambda(\mathcal{B}) \leqslant \Lambda(\mathcal{A})$.
\end{definition}
Simply put, if $\Lambda$ is monotonic non-increasing, its output cannot increase when new elements are added to its current input. Let us now introduce supermodularity.

\begin{definition}[Supermodular set function]
Let $\Lambda$ be a set function defined as $\Lambda: 2^\mathcal{\mathcal{G}} \mapsto \R$, where $2^\mathcal{G}$ is the power set of the finite set $\mathcal{G}$. $\Lambda$ is a supermodular function of $\mathcal{G}$ if and only if for every $\mathcal{A}$ and $\mathcal{B}$ such that $\mathcal{A} \subseteq \mathcal{B} \subseteq \mathcal{G}$ and every $x \in \mathcal{G} \setminus \mathcal{B}$, we have that 
$$
\Lambda(\mathcal{A}) - \Lambda(\mathcal{A} \cup \{x\}) \geqslant \Lambda(\mathcal{B}) - \Lambda(\mathcal{B} \cup \{x\}).
$$
\end{definition}

In other words, if $\Lambda$ is supermodular and it represents some notion of cost, the utility of adding an element to its current input is always greater or equal than the utility of adding the same element to a superset of its current input.

It is fundamental to understand that none of these two properties implies the other. Therefore, a supermodular function can be non-monotonic.

\section{PROBLEM FORMULATION}
\label{sec:problem_formulation}
In this work, we study traffic on transport networks as flow traversing a graph. We denote each commodity (flow origin-destination pair) as a trip. The source and destination represent the extremities of the trip, and the demand represents the number of vehicles per time unit wishing to go from that origin to that destination. The edge cost represents the travel time. \textit{From this point on we will use the terms commodity and trip interchangeably}.

Given an oriented graph $G=(\mathcal{V},\mathcal{E})$, where $\mathcal{V}$ denotes the node set and $\mathcal{E} \subseteq \mathcal{V}\times\mathcal{V}$ denotes the edge set, and a set of trips as $\mathcal{M} = \{(s^m,t^m,d^m)\}_m$, where $s^m \in \mathcal{V}$ is the origin of trip $m$, $t^m \in \mathcal{V}$ is the destination, and $d^m \in \R_{>0}$ is the trip demand. We associate to each edge $(i,j) \in \mathcal{E}$ a cost $c_{ij} \in \R_{>0}$ representing the travel time and a capacity $u_{ij} \in \R_{\geqslant 0}$ representing the maximum flow sustainable. Note that, in this formulation, we allow $c_{ij} \neq c_{ji}$. The cost can be a scalar or a function, depending on the routing formulation (see Section~\ref{sec:routing}). Each node $i \in \mathcal{V}$ has a demand $d_i^m \in \mathbb{R}, \forall m \in \mathcal{M}$.  We set the network demands $d_i^m$ as indicated by:

\begin{equation*}
\forall i \in \mathcal{V},\; \forall m \in \mathcal{M} \quad d_i^m = 
\left\{\begin{matrix} -d^m & \quad\text{if }i = s^m
\\   d^m & \quad\text{if }i = t^m
\\   0 & \quad\text{otherwise.}
\end{matrix}\right.
\end{equation*}

We observe that the total demand of the transport network is equal to zero $\sum_{i \in \mathcal{V}}\sum_{m \in \mathcal{M}}d_i^m = 0$. We represent the vehicle flow of trip $m \in \mathcal{M}$ on edge $(i,j) \in \mathcal{E}$ as $x_{ij}^m \in \R_{\geqslant 0}$, where $\sum_{m \in \mathcal{M}}x_{ij}^m = x_{ij}$ is the total flow on edge $(i,j)$. A summary of the notation used in the paper is available in~\ref{sec:notation}.

We are interested in analysing the properties of routing under path additions. Therefore, to complement our problem formulation, in the next two sections we introduce the routing paradigms used (Sec.~\ref{sec:routing}) and how we formalise path additions (Sec.~\ref{sec:path_additions}).

\section{BACKGROUND ON COORDINATED ROUTING}
\label{sec:routing}
We now introduce the three coordinated routing formulations considered in this work: MC (Sec.~\ref{sec:mcmc}), SO (Sec.~\ref{sec:so}), and UE (Sec.~\ref{sec:ue}). We refer to all of them as coordinated as they require a central authority to compute routing, but we furtherly denote MC and SO as \textit{cooperative}, as vehicles are routed for the system's benefit and UE as \textit{self-interested}, as vehicles are routed according to their own interest.
\subsection{Minimum cost multi-commodity network flow problem}
\label{sec:mcmc}
The minimum cost multi-commodity network flow problem (MC) is a classic flow problem~\cite{ahuja1988network}. It is concerned with capacitated networks in which a constant cost is associated to every edge. Given a set of trips, each with a flow demand to be routed from an origin node to a destination node, we are interested in finding the optimal flow distribution in the network such that the total cost is minimised. The shortest path problem and the maximum flow problem are special cases of the MC problem. 

The routing problem consists in determining the optimal agent distribution on the network edges in order to route all the demand $d^m$ from $s^m$ to $t^m$, $\forall m \in \mathcal{M}$ while minimising the total cost. The problem is formalised as follows:

\begin{subequations}
\begin{equation}
    \min_{\{x^m_{ij}\}}{\sum_{m \in \mathcal{M}}\sum_{(i,j) \in \mathcal{E}} x_{ij}^m c_{ij}}
    \label{eq:objective2}
\end{equation}
 s.t.
\begin{equation}
     \sum_{m \in \mathcal{M}}x_{ij}^m \leqslant u_{ij} \quad \forall (i,j) \in \mathcal{E}
    \label{eq:capacity2}
\end{equation}
\vspace{2pt}
\begin{equation}
x_{ij}^m \geqslant 0 \quad \forall (i,j) \in \mathcal{E},\;\forall m \in \mathcal{M}
\label{eq:non_zero2}
\end{equation}
\vspace{2pt}
\begin{equation}
    \sum_{\left \{j:(j,i) \in \mathcal{E} \right \}} x_{ji}^m -
    \sum_{\left \{j:(i,j) \in \mathcal{E} \right \}} x_{ij}^m = d_i^m \quad
    \forall i \in \mathcal{V},\; \forall m \in \mathcal{M},
    \label{eq:conservation2}
\end{equation}
\end{subequations}
where the objective function \eqref{eq:objective2} minimises the total travel time over all trips. Constraint \eqref{eq:capacity2} ensures that the sum of the flows from different trips on each edge does not exceed the maximum capacity. Constraint \eqref{eq:non_zero2} ensures that no flow is negative. Constraint \eqref{eq:conservation2} ensures flow conservation.
    
\textbf{Path formulation.}
According to the flow decomposition theorem, first introduced by \cite{ford2015flows}, we can reformulate the MC problem using the distribution of flow over paths instead of edges. This is valid under the assumption that for every trip there exist no negative cost cycles in our graph.

For every trip $m \in \mathcal{M}$, let $P^m$ denote the set of all possible paths\footnote{A path is a is a walk in which all vertices (and therefore also all edges) are distinct.} from $s^m$ to $t^m$ and $P = \bigcup_{m \in \mathcal{M}} P^m$ the union of such sets. The constant cost $c_p$ of a path is defined as the sum of the constant cost of its edges $c_p = \sum_{(i,j) \in p}c_{ij}$. We also define as $x_p$ the flow on path $p \in P^m$ and as $\mathbf{x}$ the vector of all path flows $x_p,\; \forall p \in P$. The path formulation goes as follows:

\begin{subequations}
\begin{equation}
    \min_{\mathbf{x}}{\sum_{p \in P} x_p c_p }
    \label{eq:objective3}
\end{equation}
 s.t.
\begin{equation}
     \sum_{p \in P:(i,j) \in p}x_p \leqslant u_{ij} \quad \forall (i,j) \in \mathcal{E}
    \label{eq:capacity3}
\end{equation}
\vspace{2pt}
\begin{equation}
x_p \geqslant 0 \quad \forall p \in P
\label{eq:non_zero3}
\end{equation}
\vspace{2pt}
\begin{equation}
    \sum_{p \in P^m} x_p = d^m \quad
    \forall m \in \mathcal{M},
    \label{eq:conservation3}
\end{equation}
\end{subequations}
where the objective function and the constraints map exactly to the ones illustrated in the edge formulation of MC.

\subsection{Flow-dependant cost}
\label{sec:flow-dependant_cost}
In reality, it is often not possible to map the travel time of road segments to a constant cost. In fact, \textit{the travel time required to traverse a road is highly dependent on its congestion}. It can be modeled by a continuous monotonic increasing function mapping the current flow on an edge to travel time. In this work, we will consider two popular travel time functions used in traffic research, but our theoretical results remain valid for any reasonable convex travel time function.

\subsubsection{Greenshields affine velocity function}
\label{sec:greenshields}

This function was introduced by Greenshields in 1935~\cite{greenshields1935study} and still remains a useful mathematical model thanks to its simplicity. It is  well-known in traffic simulation research~\cite{siebel2006fundamental, work2010traffic}. It is defined as follows:

\begin{equation}
    c_{ij}(x_{ij}) = \frac{l_{ij}}{v_{ij}^{\mathrm{max}}\left (1-\frac{x_{ij}}{u_{ij}}\right)},
    \label{eq:greenshields}
\end{equation}
where $v_{ij}^{\mathrm{max}}$ is the maximum velocity allowed on edge $(i,j)$ and $l_{ij}$ is the length of $(i,j)$. Note that the input variable $x_{ij}$ has to be a feasible assignment and thus $0 \leqslant x_{ij} \leqslant u_{ij}$. The function models a convex hyperbolic relation between congestion and travel time.

\subsubsection{The Bureau of Public Roads (BPR) cost function}
\label{sec:bpr}
Another frequently considered convex latency function is the Bureau of Public Roads (BPR) cost function~\cite{sheffi1985equilibrium, branston1976link}.  The function is computed as

\begin{equation}
    c_{ij}(x_{ij}) = c_{ij}^0 \left (1 + \alpha_{ij} \left (\frac{x_{ij}}{u_{ij}} \right )^{\beta_{ij}} \right),
    \label{eq:brp}
\end{equation}
where $c^0_{ij}$ and $u_{ij}$ are the free-flow time and capacity of link $(i,j)$, respectively, and $\alpha_{ij}$ and $\beta_{ij}$ are shape parameters which can be calibrated to data. By tuning these parameters a wide variety of functions can be represented, including the ones discussed in Section~\ref{sec:braess} used by Braess to exacerbate his paradox~\cite{braess2005paradox}.
\subsection{System optimal (SO) routing}
\label{sec:so}
Now we can use the cost functions just introduced to extend our formulation of the MC problem. Each edge $(i,j) \in \mathcal{E}$ is assigned the flow-dependant cost function $c_{ij}(x_{ij})$. The vector of edge cost functions $c_{ij}(x_{ij})$ is noted as $\mathbf{c}$. We call the triple $(G,\mathcal{M}, \mathbf{c})$ an instance. The cost of a path $p \in P$ with respect to a flow $x_{p}$ is defined as the sum of the cost of the edges in the path $c_p(x_p) = \sum_{(i,j) \in p}c_{ij}(x_p)$.

The optimal flows are then characterised by the following System Optimal (SO)~\cite{roughgarden2002bad} routing problem:

\begin{subequations}
\label{eq:so}
\begin{equation}
    \min_{\mathbf{x}}{\sum_{p \in P} x_p c_p(x_p)}
    \label{eq:sor_objective}
\end{equation}
 s.t.
\begin{equation}
x_p \geqslant 0 \quad \forall p \in P
\label{eq:non_zero_sor}
\end{equation}
\vspace{2pt}
\begin{equation}
    \sum_{p \in P^m} x_p = d^m \quad
    \forall m \in \mathcal{M},
    \label{eq:conservation_sor}
\end{equation}
\end{subequations}
where constraint~\eqref{eq:non_zero_sor} states that all flows must be non-negative and constraint~\eqref{eq:conservation_sor} imposes the conservation of flow for each trip. Note that the SO problem corresponds to the MC problem with flow-dependant travel time functions and without the strict edge capacity constraints.

Since the local and global optima of a convex function on a convex set coincide~\cite{peressini1988mathematics}, we know that a locally optimal solution for SO is also globally optimal whenever the objective function~\eqref{eq:sor_objective} is convex. By construction, we consider only convex cost functions, like the ones introduced in Section~\ref{sec:flow-dependant_cost}. The path cost $c_p(x_p)$ and the SO objective are thus also convex, since a non-negative weighted sum of convex functions is still convex. From~\cite{roughgarden2002bad} we know that a flow is locally (and globally) optimal if and only if moving flow from one path to another can only increase the global cost~\eqref{eq:sor_objective}. That is, we expect to have optimality when the marginal benefit of decreasing flow on a $s^m-t^m$ path is at most the marginal cost of increasing flow on another $s^m-t^m$ path. The following lemma formalises this intuition. Let us denote the effective cost $k_{ij}(x_{ij}) = x_{ij} c_{ij}(x_{ij})$. We write the derivative of $k_{ij}$ as $k_{ij}'$ and define $k_p'(x_p) = \sum_{(i,j) \in p}k_{ij}'(x_p)$. We then have:

\begin{lemma}[\cite{beckmann1956studies, dafermos1969traffic, roughgarden2002bad}]
The flow vector $\mathbf{x}$ computed by SO routing on an instance  $(G,\mathcal{M},\mathbf{c})$ is optimal if and only if for every $m \in \mathcal{M}$ and any pair $p_1,p_2 \in P^m$ with $x_{p_1} > 0$, we have $k_{p_1}'(x_{p_1}) \leqslant k_{p_2}'(x_{p_2})$.
\label{lemma:sor}
\end{lemma}

The solution to the SO routing problem thus characterises the scenario where all users choose their paths in order to minimise the total system travel time or, in other words, the social cost.
\subsection{User equilibrium (UE) routing}
\label{sec:ue}
By solving the SO problem we optimise for total system travel time. On the other hand, some agents could experience longer travel times than others on the same trip for the system's benefit. Due to this characteristic, we define SO as \textit{unfair}.

A fair assignment is the User Equilibrium (UE). It induces a Nash equilibrium in which no user can improve their travel time (cost) by choosing a different route. It is also known as Wardrop equilibrium~\cite{wardrop1952road}. This equilibrium emerges when we consider the routing problem as a non-cooperative game, where each agent, seen as a fraction of the flow, is routed on the fastest route available according to the current traffic conditions. In this sense, agents routed according to UE can be defined as self-interested.

\begin{lemma}[\cite{roughgarden2002bad}]
A feasible flow $\mathbf{x}$ for instance $(G,\mathcal{M},\mathbf{c})$ is at UE if and only if for every $m \in \mathcal{M}$ and any pair $p_1,p_2 \in P^m$ with $x_{p_1} > 0$, we have $c_{p_1}(x_{p_1}) \leqslant c_{p_2}(x_{p_2})$.
\label{lemma:uer}
\end{lemma}
In particular, if a flow is at UE, all the paths of trip $m$ for which $c_p > 0$ have the same cost $C_m$.

\begin{lemma}[\cite{roughgarden2002bad}]
If a flow is at UE for instance $(G,\mathcal{M}, \mathbf{c})$, then the total cost is computed as

\begin{equation}
    \sum_{p \in P} x_p c_p(x_p) = \sum_{m \in \mathcal{M}} d^m C_m.
    \label{eq:user_time}
\end{equation}
\end{lemma}
This means that all agents on the same trip will arrive at the same time to their destination under UE. The convexity guarantees of SO are valid also in this case and, in particular, we have strong existence guarantees for the equilibrium.

\begin{lemma}[\cite{roughgarden2002bad}]
Given an instance $(G,\mathcal{M}, \mathbf{c})$, if there exists a feasible flow assignment, there always exists a feasible flow assignment at UE. Furthermore, given two assignments at UE, they will have the same total cost (travel time) (shown in Eq.~\eqref{eq:user_time}).
\end{lemma}

Lastly, let us define the UE routing problem as follows:

\begin{subequations}
\label{eq:ue}
\begin{equation}
    \min_{\mathbf{x}}{\sum_{(i,j) \in \mathcal{E}} \int_{0}^{x_{ij}}c_{ij}(\tau) d\tau}
    \label{eq:uer_objective}
\end{equation}
 s.t.
\begin{equation}
x_p \geqslant 0 \quad \forall p \in P
\label{eq:non_zero_uer}
\end{equation}
\vspace{2pt}
\begin{equation}
    \sum_{p \in P^m} x_p = d^m \quad
    \forall m \in \mathcal{M}
    \label{eq:conservation_uer}
\end{equation}
\begin{equation}
    x_{ij} = \sum_{p \in P : (i,j) \in p}x_p \quad \forall (i,j) \in \mathcal{E},
    \label{eq:paths_uer}
\end{equation}
\end{subequations}
where the constraints are the same as in the SO case, with the addition of~\eqref{eq:paths_uer}, which states that the flow on an edge is equal to the sum of the flows of the paths passing through it. This addition is necessary as in the objective~\eqref{eq:uer_objective} we sum over edges instead of paths.
\subsection{Differences between SO and UE routing}

 The SO and UE problems are different. In particular, when considering the UE formulation, we trade-off total travel time for fairness and self-interest.

However, by looking at Lemmas~\ref{lemma:sor} and~\ref{lemma:uer}, we can notice some similarities.
The striking similarity between the characterizations of optimal solutions to the SO and UE problems was noticed early on~\cite{beckmann1956studies}, and provides an interpretation of a system optimal flow as a flow at UE with respect to a different set of edge cost functions~\cite{roughgarden2002bad}. To make this relationship precise, denote the marginal cost of increasing flow on edge $(i,j)$ by $c_{ij}^*(x_{ij}) = k_{ij}'(x_{ij}) = (x_{ij} c_{ij}(x_{ij}))' = c_{ij}(x_{ij}) + x_{ij} c_{ij}'(x_{ij})$. Lemmas~\ref{lemma:sor} and~\ref{lemma:uer} yield the following corollary.

\begin{corollary}[\cite{roughgarden2002bad}]
Let $(G, \mathcal{M}, \mathbf{c})$ be an instance and $\mathbf{c}^*$ be the vector of marginal cost functions defined as above. Then, a flow feasible for $(G, \mathcal{M}, \mathbf{c})$ is SO if and only if it is at UE for the instance $(G,\mathcal{M}, \mathbf{c}^*)$.
\end{corollary}

Both routing strategies represent convex nonlinear optimisation problems, which can be solved by traditional gradient-based techniques such as the Frank-Wolfe algorithm~\cite{frank1956algorithm, leblanc1975efficient}. The cost trade-off in total travel time between SO and UE is known as ``The price of anarchy"~\cite{roughgarden2005selfish, roughgarden2003price}. Note that, in a network with constant costs, the optimal solution of SO and UE is the same as that of the MC problem (assuming we ignore the capacity constraint~\eqref{eq:capacity2}).

\begin{example}
Let us consider the example shown in Fig.~\ref{fig:so_vs_ue}, first introduced by Pigou~\cite{pigou2017economics} to illustrate the difference between SO and UE routing. In this simple scenario we consider only one trip, from $s$ to $t$, with a demand of $1$ unit of flow. In the network there are only two paths, represented as edges. The travel time of such paths is indicated by $c_1(x_{1})$ and $c_2(x_{2})$, which are functions of the flow. The flow routed on the upper road is indicated as $x_1$, while the flow on the lower road is $x_2$.

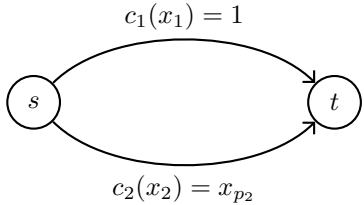
\begin{figure}[ht!]
\centering

  \begin{tikzpicture}[node distance={40mm}, thick, 
    node/.style = {draw, circle, minimum size=7mm},
    edge_label/.style = {midway, above, sloped, color=black},
    edge/.style = {->, color=black}
    ] 
    
    \begin{scope}[>={Straight Barb[black]}]
    
    \node[node] (1) {$s$}; 
    \node[node] (2) [right of=1] {$t$};
    \draw[edge] (1) to  [out=45,in=135,looseness=0.8] node[edge_label] {$c_1(x_{1}) = 1$} (2);
    \draw[edge] (1) to  [out=-45,in=-135,looseness=0.8] node[midway, below, sloped, color=black] {$c_2(x_{2}) = x_{p_2}$} (2);
    \end{scope}

\end{tikzpicture} 
\caption{Pigou's example to illustrate the difference between UE and SO. It presents one origin $s$ and one destination $t$ and a flow demand of 1 that has to be distributed along paths 1 and 2, with travel time costs respectively $c_1(x_1)$ and $c_2(x_2)$. While SO distributes the flow equally among the two paths, UE routes all the flow on the lower path to guarantee equal costs on the two paths.}
\label{fig:so_vs_ue}
\end{figure}

If the users are routed according to UE, the flow will distribute such that all the users experience the same travel time. This leads to a UE where all the flow is routed on the lower road ($x_1 = 0, x_2=1$), yielding $c_1(x_1) = c_2(x_2) = 1$. The total system travel time in this case is $x_1 c_1(x_1) + x_2 c_2(x_2) = 1$. On the other hand, if we optimise for the SO, the flow distributes equally on the two roads ($x_1 = 0.5, x_2=0.5$), yielding a total travel time of $x_1 c_1(x_1) + x_2 c_2(x_2) = 0.75$. This example is summarised in Table~\ref{tab:UE_vs_SO}.

\begin{table}[H]
\caption{UE and SO routing results on the network in Fig.~\ref{fig:so_vs_ue}.}
\centering
\begin{tabular}{r c c c p{.5\linewidth}}
\toprule
 Routing & $x_1$ & $x_2$  & Total travel time \\
 &&& $x_1 c_1(x_1) + x_2c_2(x_2)$\\
\midrule
User Equilibrium & 0 & 1 & 1\\
System Optimal & 0.5 & 0.5 & 0.75 \\
\bottomrule
\end{tabular}

\label{tab:UE_vs_SO}
\end{table}
\end{example}

\section{PATH ADDITIONS}
\label{sec:path_additions}
Now that we have introduced the routing problems of interest, we turn our attention to the transport network on which vehicles are routed, represented as a graph. We are interested in developing a formulation for the incremental construction of the graph which enables us to formally reason about the properties of path additions. We introduce two concepts: the \textit{trip spanning tree} and the \textit{trip path graph}. These two components allow us to simply model path additions in a modular framework, which is also able to represent the example introduced by Braess'~\cite{braess2005paradox}.

We consider graphs of the type introduced in Section~\ref{sec:problem_formulation}. All the graphs we will consider, will be subgraphs of a \textit{template graph} $G_\mathcal{T}$. This template graph is used to limit the space of possible path additions. In this work we consider the template graph to be a square grid lattice with bidirectional edges, but any template choice is equally valid.

As we are interested in analysing the properties of path additions to a transport network, we have to start by defining an initial feasible graph which we aim to extend. 
Therefore, we start by introducing the concept of \textit{trip spanning tree}.

\begin{definition}[\textbf{Trip spanning tree}]
    Given a template graph $G_{\mathcal{T}}$ and a set of trips $\mathcal{M} = \{(s^m,t^m,d^m)\}_m$, a trip spanning tree is a directed graph $G_{\mathcal{I}}=(\mathcal{V}_{\mathcal{I}},\mathcal{E}_{\mathcal{I}})$, subgraph of $G_{\mathcal{T}}$, with the following four properties:
    \begin{enumerate}
        \item $s^m,t^m \in \mathcal{V}_\mathcal{I} ,\; \forall m \in \mathcal{M}$
        \item $|P^m| = 1 ,\; \forall m \in \mathcal{M}$
        \item $ 
        \forall n \in \mathcal{V}_\mathcal{I},\;
        \exists m \in \mathcal{M},\; p \in P^m: n \in p$
        \item $ \sum_{m \in \mathcal{M}}\sum_{p \in P^m:(i,j) \in p} d^m \leqslant u_{ij} \quad \forall (i,j) \in \mathcal{E}_\mathcal{I}$.
    \end{enumerate}

\end{definition}

 Property (1) states that a trip spanning tree must contain the source and destination of every trip. Property (2) states that for each origin-destination pair there exists exactly one path connecting them in the trip spanning tree. Property (3) states that every node in the trip spanning tree must be part of a path connecting one trip's origin to its destination. Property (4) states that the capacity of the edges in the trip spanning tree must be sufficient to guarantee a feasible solution to the MC problem. These properties are designed to make the trip spanning tree an initial feasible graph for the routing problem. Note that the trip spanning tree is not unique with respect to the problem instance.

Let us now define the possible path additions through the concept of \textit{trip path graph}.

\begin{definition}[\textbf{Trip path graph}]
     Given a template graph $G_{\mathcal{T}}$ and a set of trips $\mathcal{M} = \{(s^m,t^m,d^m)\}_m$, we denote a trip path graph, with respect to trip $m \in \mathcal{M}$, as a graph $G_{m_x} = (\mathcal{V}_{m_x},\mathcal{E}_{m_x})$, subgraph of $G_{\mathcal{T}}$, with the following properties:
    \begin{enumerate}
        \item $s^m,t^m \in \mathcal{V}_{m_x}$
        \item $|P^m| = 1$
        \item $ 
        \forall n \in \mathcal{V}_{m_x},\; p \in P^m: n \in p$.
    \end{enumerate}
\end{definition}
 
 A trip path graph is thus a path graph which contains only a path from one trip's source to its destination. The number of possible trip path graphs for trip $m$ is dependent on the template graph $G_{\mathcal{T}}$ and is indexed by $x$. Property (1) states that the origin and destination of such trip must belong to the trip path graph. Property (2) states that the trip path graph must contain only one path between these two. Property (3) states that all the nodes belonging to the trip path graph must belong to such path. Therefore, a trip path graph is a trip spanning tree with respect to only one trip, without the guarantees given by Property (4) of a trip spanning tree. Note that the Braess paradox (Fig.~\ref{fig:braess}) can be formulated according to this framework. Here a trip path graph $(s,v,w,t)$ is added to a trip spanning tree that has already undergone one path addition.

 \begin{figure}[ht!]
\centering
\begin{subfigure}{0.5\columnwidth}
\centering
\scalebox{0.5}{%
  \begin{tikzpicture}[node distance={40mm}, thick, 
    node/.style = {draw, circle, minimum size=0.75cm},
    edge/.style = {<->}
    ] 
    
    \node[node] (9)  [fill=green!80]{$s^1$};
    \node[node] (10) [right=1.6cm of 9]{};
    \node[node] (2) [above=1.3cm of 10]{};
    \node[node] (3) [right=1.6cm of 2]{};
    \node[node] (4) [right=1.6cm of 3, fill=green!80]{$t^1$};
    \node[node] (6) [above=1.3cm of 2,fill=SkyBlue!80]{$s^2$};
    \node[node] (11) [below=1.3cm of 3]{};

    \node[node] (14) [below=1.3cm of 11, fill=SkyBlue!80]{$t^2$};

    \begin{scope}[>={Straight Barb[black]}, every edge/.style={draw, color=black,very thick}]
    \draw[->] (9) edge (10);
    \draw[->] (2) edge (3);
    \draw[->] (3) edge (4);
    \draw[->] (6) edge (2);
    \draw[<->] (2) edge (10);
    \draw[<-] (14) edge (11);
    \draw[->] (10) edge (11);
    \end{scope}
    \begin{scope}[>={Straight Barb[orange]}, every edge/.style={draw, color=orange,very thick}]
    ;

    \end{scope}
    
\end{tikzpicture} 
}
\caption{Trip spanning tree $G_\mathcal{I}$}
\label{fig:minimal}
\end{subfigure}%
\begin{subfigure}{0.5\columnwidth}
\centering
\scalebox{0.5}{%
  \begin{tikzpicture}[node distance={40mm}, thick, 
    node/.style = {draw, circle, minimum size=0.75cm},
    edge/.style = {<->}
    ] 
    
    \node[node] (9)  [fill=green!80]{$s^1$};
    \node[node] (10) [right=1.6cm of 9]{};
    \node[node] (2) [draw=none, above=1.3cm of 10]{};
    \node[node] (3) [right=1.6cm of 2]{};
    \node[node] (4) [right=1.6cm of 3, fill=green!80]{$t^1$};
    \node[node] (6) [draw=none, above=1.3cm of 2]{};
    \node[node] (11) [below=1.3cm of 3]{};

    \node[node] (14) [draw=none, below=1.3cm of 11]{};

    \begin{scope}[>={Straight Barb[orange]}, every edge/.style={draw, color=orange,very thick}]
    \draw[->] (9) edge (10);
    \draw[->] (3) edge (4);
    \draw[->] (10) edge (11);
    \draw[<-] (3) edge (11);
    \end{scope}
    
\end{tikzpicture} 
}
\caption{Trip path graph $G_{1_1}$}
\label{fig:addition}
\end{subfigure}
\begin{subfigure}{\columnwidth}
\centering
\vspace{5pt}
\scalebox{0.8}{%
  \begin{tikzpicture}[node distance={40mm}, thick, 
    node/.style = {draw, circle, minimum size=0.75cm},
    edge/.style = {<->}
    ] 
    
    \node[node] (9)  [fill=green!80]{$s^1$};
    \node[node] (10) [right=1.6cm of 9]{};
    \node[node] (2) [above=1.3cm of 10]{};
    \node[node] (3) [right=1.6cm of 2]{};
    \node[node] (4) [right=1.6cm of 3, fill=green!80]{$t^1$};
    \node[node] (6) [above=1.3cm of 2,fill=SkyBlue!80]{$s^2$};
    \node[node] (11) [below=1.3cm of 3]{};

    \node[node] (14) [below=1.3cm of 11, fill=SkyBlue!80]{$t^2$};

    \begin{scope}[>={Straight Barb[black]}, every edge/.style={draw, color=black,very thick}]
    \draw[->] (9) edge (10);
    \draw[->] (2) edge (3);
    \draw[->] (3) edge (4);
    \draw[->] (6) edge (2);
    \draw[<->] (2) edge (10);
    \draw[<-] (14) edge (11);
    \draw[->] (10) edge (11);
    \end{scope}
    \begin{scope}[>={Straight Barb[orange]}, every edge/.style={draw, color=orange,very thick}]
    ;
    \draw[<-] (3) edge (11);

    \end{scope}
    
\end{tikzpicture} 
}
\caption{Graph union $G_{\mathcal{I} \oplus 1_1}$}
 \label{fig:result}
\end{subfigure}
\caption{A simple transport network with two trips to illustrate graph unions. Here we have a trip spanning tree \subref{fig:minimal} and a trip path graph for trip 1 \subref{fig:addition} that are being unified in \subref{fig:result}.}
\label{fig:example}
\end{figure}
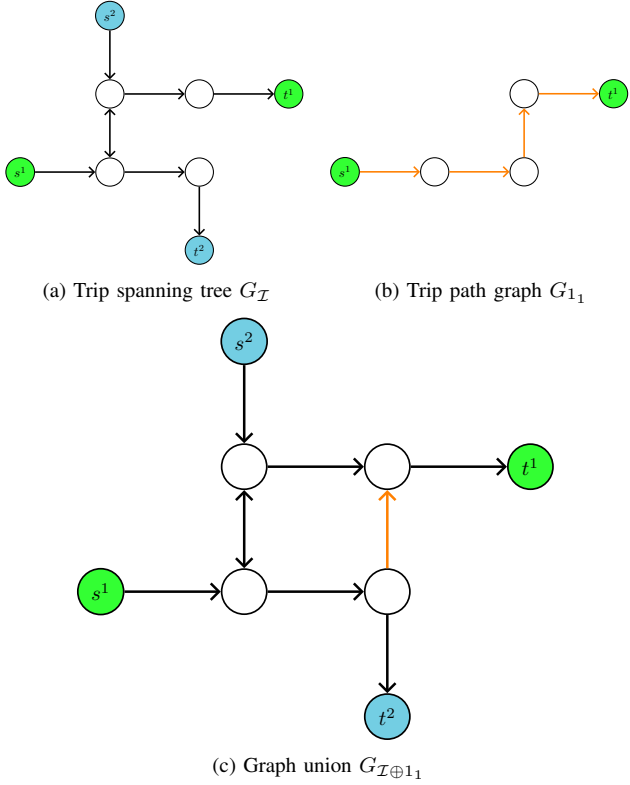

The trip spanning tree and the trip path graphs allow us to model the subsequent addition of paths to an initial graph. We describe this through graph unions. We denote the graph obtained by adding a trip path graph $G_{m_x} = (\mathcal{V}_{m_x},\mathcal{E}_{m_x})$ to the trip spanning tree $G_{\mathcal{I}}=(\mathcal{V}_{\mathcal{I}},\mathcal{E}_{\mathcal{I}})$ as $G_{\mathcal{I} \oplus {m_x}}=(\mathcal{V}_{\mathcal{I}} \cup \mathcal{V}_{m_x},\mathcal{E}_{\mathcal{I}} \cup \mathcal{E}_{m_x})$. Note that, by this addition, the trip path graph can introduce multiple new nodes and edges not previously contained in the trip spanning tree and thus also multiple new paths for each trip. The number of new paths introduced is limited only by the template graph $G_\mathcal{T}$ (remembering that $G_\mathcal{I} \subseteq G_\mathcal{T}$ and $G_{m_x} \subseteq G_\mathcal{T}$).
 
 We illustrate the concept of trip spanning tree and trip path graph addition with an example, shown in Fig.~\ref{fig:example}. In this simple example we have only two trips $\mathcal{M} = \{(s^1,t^1,d^1),(s^2,t^2,d^2)\}$. We instantiate a possible trip spanning tree $G_{\mathcal{I}}$ (Fig.~\ref{fig:minimal}) and a possible trip path graph $G_{1_1}$ (Fig.~\ref{fig:addition}) for trip 1. The resulting graph union $G_{\mathcal{I} \oplus 1_1}$ is shown in Fig.~\ref{fig:result}. Despite $G_{1_1}$ has four edges, only one new edge is added to the graph, since the other edges were already present. Furthermore, note that, in this case, $G_{1_1}$ adds one path for trip 1 and no paths for trip 2. This is not always the case as adding a trip path graph for trip $m$ could introduce new paths for other trips as well.

 Let us introduce some further notation related to graph unions. We consider a graph $G_\mathcal{A}=(\mathcal{V}_\mathcal{A},\mathcal{E}_\mathcal{A})$, obtained after $n \geqslant 0$ trip path graph additions to a trip spanning tree, and a trip path graph $G_{m_x}=(\mathcal{V}_{m_x},\mathcal{E}_{m_x})$. We denote with $P_\mathcal{A}^m$ the set of all paths for trip $m$ in $G_\mathcal{A}$, with $P_{{m_x}|\mathcal{A}}^m$ the set of all paths added by $G_{m_x}$ to $G_\mathcal{A}$ for trip $m$, and with $P_{{m_x} \oplus \mathcal{A}}^m $ the set of all paths for trip $m$ in $G_{{m_x} \oplus \mathcal{A}}$. Note that $P_{{m_x}|\mathcal{A}}^m$ is defined as $P_{{m_x}|\mathcal{A}}^m = P_{{m_x} \oplus \mathcal{A}}^m \setminus P_{\mathcal{A}}^m$. From this definition we obtain that $P_{{m_x} \oplus \mathcal{A}}^m = P_{{m_x}|\mathcal{A}}^m \cup P_{\mathcal{A}}^m$, with $P_{{m_x}|\mathcal{A}}^m \cap P_{\mathcal{A}}^m = \varnothing$. We also define $P_\mathcal{A} = \bigcup_{m \in \mathcal{M}} P_\mathcal{A}^m$.
 
 Let us look at another example for this notation. Suppose that there exists only one trip $\mathcal{M} = \{(s^1,t^1,d^1)\}$. In Figure~\ref{fig:example2}, we can see a trip spanning tree $G_\mathcal{I}$ for trip 1 (shown in black) and a trip path graph $G_{1_1}$ for the same trip (shown in orange). Therefore, we will have $|P_\mathcal{I}^1| = |P_{1_1}^1| = 1$. These two sets will contain the black and the orange path respectively. The set $P_{1_1 \oplus \mathcal{I}}^1 = P_{1_1|\mathcal{I}}^1 \cup P_{\mathcal{I}}^1$  will contain four paths: three of which come from the set $P_{1_1|\mathcal{I}}^1$, which contains the path in $P_{1_1}^1$ ($(s^1,w,r,v,g,h,t^1)$) and two new induced paths  ($(s^1,v,g,h,t^1), (s^1,w,r,v,t^1)$), and the last being the original path present in $P_\mathcal{I}^1$ ($(s^1,v,t^1)$).
 
\begin{figure}[ht]
\centering
\begin{tikzpicture}[node distance={40mm}, thick, 
    node/.style = {draw, circle, minimum size=0.87cm},
    edge/.style = {<->}
    ] 
    
    \node[node] (1) [fill=green!80]{$s^1$}; 
    \node[node] (2) [right=1.6cm of 1]{$v$};
    \node[node] (3) [right=1.6cm of 2, fill=green!80]{$t^1$};
    \node[node] (4) [above=1.3cm of 1]{$w$}; 
    \node[node] (5) [above=1.3cm of 2]{$r$};
    \node[node] (6) [below=1.3cm of 2]{$g$};
    \node[node] (7) [below=1.3cm of 3]{$h$};

    \begin{scope}[>={Straight Barb[black]}, every edge/.style={draw, color=black,very thick}]
    \draw[->] (1) edge (2);
    \draw[->] (2) edge (3);
    \end{scope}
    \begin{scope}[>={Straight Barb[orange]}, every edge/.style={draw, color=orange,very thick}]
    ;
    \draw[->] (1) edge (4);
    \draw[->] (4) edge (5);
    \draw[->] (5) edge (2);
    \draw[<-] (6) edge (2);
    \draw[->] (6) edge (7);
    \draw[->] (7) edge (3);
    \end{scope}
    
\end{tikzpicture}
\caption{Illustration of the addition of a trip path graph. In this figure we have a trip spanning tree $G_\mathcal{I}$ for trip 1 (shown in black) to which we add a trip path graph $G_{1_1}$ for the same trip (shown in orange). We see how, by only adding one trip path graph to the trip spanning tree, we obtain four total paths for trip 1, of which: one is the original one present in the trip spanning tree ($(s^1,v,t^1)$), one is present in $G_{1_1}$ ($(s^1,w,r,v,g,h,t^1)$), and the remaining two ($(s^1,v,g,h,t^1), (s^1,w,r,v,t^1)$) are created from the addition.}
\label{fig:example2}
\end{figure}
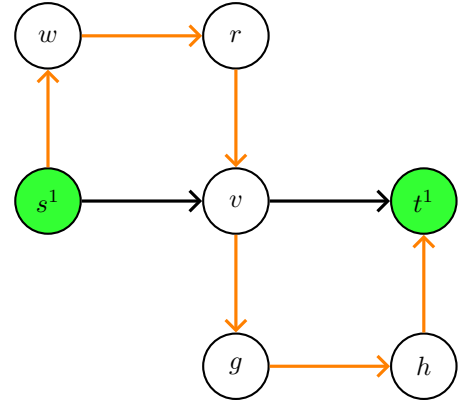

\section{PROPERTIES OF PATH ADDITIONS}
\label{sec:properties}
Having defined the problem of path additions, we turn our attention to its properties and, in particular to monotonicity and supermodularity.

Let us start by introducing the functions we will analyse. Given a graph template $G_\mathcal{T}$ with constant edge costs, a set of trips $\mathcal{M}$, and a trip spanning tree $G_\mathcal{I}$, we denote as $\mathcal{G}$ the finite set of all possible trip path graphs. Note that $\mathcal{G} = \bigcup_{m \in \mathcal{M}}{\mathcal{G}_m}$ is the union of the disjoint sets of trip path graphs for each trip, where $\mathcal{G}_m = \{G_{m_1}, G_{m_2}, \dots, G_{m_x}, \dots\}$ is the set of all trip path graphs for trip $m$. We then define $\Lambda_{\mathrm{MC}}: 2^\mathcal{\mathcal{G}} \mapsto \R_{>0}$ which takes as input a subset of $\mathcal{G}$, adds it to $G_{\mathcal{I}}$, and outputs the cost of the MC routing on the resulting graph. This function can be written as follows:

\begin{equation*}
   \Lambda_{\mathrm{MC}}(\mathcal{A}) = {\sum_{p \in P_{\mathcal{A} \oplus \mathcal{I}}} x^{\mathrm{MC}}_p c_p},
\end{equation*} 
 where $\mathbf{x}^{\mathrm{MC}}$ is the solution to the MC problem on graph $G_{\mathcal{A} \oplus \mathcal{I}}$. $\mathcal{A} \subseteq \mathcal{G}$ is a set of trip path graphs, equally identified by the union of all the graphs in it, which we denote as $G_{\mathcal{A}} = \bigcup_{G_x \in \mathcal{A}} G_x$, containing the set of paths $P_\mathcal{A}$. The trip spanning tree $G_\mathcal{I}$ guarantees that when $\Lambda_{\mathrm{MC}}$ is computed on the empty set ($\Lambda_{\mathrm{MC}}(\varnothing)$) there is still a feasible flow for each trip and each trip's source is connected to its destination.
 
 Similarly, we define $\Lambda_{\mathrm{SO}}$ and $\Lambda_{\mathrm{UE}}$, which operate on graphs with flow-dependant cost functions. We call $\mathcal{H}$ the set of all possible trip path graphs drawn from a template with flow-dependant costs. $\Lambda_{\mathrm{SO}}, \Lambda_{\mathrm{UE}}: 2^\mathcal{\mathcal{H}} \mapsto \R_{>0}$ are defined as follows:
 
\begin{subequations}
\begin{equation*}
     \Lambda_{\mathrm{SO}}(\mathcal{A}) = {\sum_{p \in P_{\mathcal{A} \oplus \mathcal{I}}} x^{\mathrm{SO}}_p c_p(x^{\mathrm{SO}}_p)}
     \label{eq:lambda_o}
\end{equation*}
\begin{equation*}
     \Lambda_{\mathrm{UE}}(\mathcal{A}) = {\sum_{p \in P_{\mathcal{A} \oplus \mathcal{I}}} x^{\mathrm{UE}}_p c_p(x^{\mathrm{UE}}_p)},
     \label{eq:lambda_u}
\end{equation*}
\end{subequations}
where $\mathbf{x}^{\mathrm{SO}}$ and $\mathbf{x}^{\mathrm{UE}}$ are the solutions to the SO and UE flow problems on graph $G_{\mathcal{A} \oplus \mathcal{I}}$.

It is important to understand that $\Lambda_{\mathrm{MC}}$ could be a particular case of $\Lambda_{\mathrm{SO}}$ with constant costs. The only difference is that $\Lambda_{\mathrm{MC}}$ is subject to the capacity constraint \eqref{eq:capacity2} of the MC problem. While this constraint is ensured to be satisfied for the trip spanning tree (Property 4), we do not enforce the same for path additions.

We can use the functions just introduced to define NDPs via path additions. The problems will consist in minimizing one of the three functions introduced (depending on the routing formulation), subject to cardinality constraints, construction budgets, or general constraints formulated as matroids. The lower level of these NDPs is encapsulated in our objective functions.

To understand whether we can design near-optimal approximation algorithms for these NDPs, we are now interested in analysing the properties of the functions just introduced.

\subsection{Monotonicity}
Let us analyse the monotonicity of $\Lambda_{\mathrm{MC}}, \Lambda_{\mathrm{SO}}, \Lambda_{\mathrm{UE}}$.

\begin{theorem}[\cite{braess2005paradox}]
    $\Lambda_{\mathrm{UE}}$ is \textbf{not} a monotonic non-increasing set function of the set $\mathcal{H}$.
\end{theorem}
\begin{proof}
The proof is provided by the Braess' paradox~\cite{braess2005paradox}. His counterexample is discussed in Section~\ref{sec:braess}.
\end{proof}
\begin{theorem}
    $\Lambda_{\mathrm{SO}}$ is a monotonic non-increasing set function of the set $\mathcal{H}$.
\label{theo:so_non_increasing}
\end{theorem}

\begin{proof}
The proof can be found in~\ref{sec:proof_so_non_inreasing}.
\end{proof}

\begin{theorem}
    $\Lambda_{\mathrm{MC}}$ is a monotonic non-increasing set function of the set $\mathcal{G}$.
\label{theo:mc_non_increasing}
\end{theorem}
\begin{proof}
The proof can be found in~\ref{sec:proof_mc_non_increasing}.
\end{proof}

The intuition behind the proofs of monotonicity is simple: when the objective is to minimise the social cost, like in MC and SO, the addition of a path can either provide a decrease in total cost or be ignored\footnote{Inconvenient additions can be ignored thanks to the fact that the trip spanning tree always grants a feasible starting solution.} by the agents. This is not the case for UE, as shown by Braess.

\subsection{Supermodularity}
The property of supermodularity of path additions can be informally explained through the following statement: `adding a path to a transport network always grants greater or equal benefit to users than adding it to a superset of that network'. Although this property may appear intuitive at first, we show that it does not hold in general, even in the case where all the agents are cooperating to minimise the social cost. As in Braess' work, we present a counterexample to support our proofs, shown in Fig.~\ref{fig:counterexample_new} and~\ref{fig:counterexample_steps}.

\begin{figure}[h!]
\centering
  \begin{tikzpicture}[node distance={40mm}, thick, 
    node/.style = {draw, circle, minimum size=0.87cm},
    edge/.style = {<->}
    ] 
    
    \node[node] (1) [fill=green!80]{$s^1$}; 
    \node[node] (2) [right=1.6cm of 1]{};
    \node[node] (3) [right=1.6cm of 2]{};
    \node[node] (4) [right=1.6cm of 3, fill=green!80]{$t^1$};
    \node[node] (5) [above=1.3cm of 1]{}; 
    \node[node] (6) [above=1.3cm of 2]{};
    \node[node] (7) [above=1.3cm of 3]{};
    \node[node] (8) [above=1.3cm of 4]{};
    \node[node] (9) [below=1.3cm of 1]{}; 
    \node[node] (10) [below=1.3cm of 2]{};
    \node[node] (11) [below=1.3cm of 3]{};
    \node[node] (12) [below=1.3cm of 4]{};
    \node[node] (13) [below=1.3cm of 10]{}; 
    \node[node] (14) [below=1.3cm of 11]{};

    \begin{scope}[>={Straight Barb[black]}, every edge/.style={draw, color=black,very thick}]
    \draw[->] (1) edge (2);
    \draw[->] (2) edge (3);
    \draw[->] (3) edge (4);
    \end{scope}
    \begin{scope}[>={Straight Barb[orange]}, every edge/.style={draw, color=orange,very thick}]
    \draw[->] (1) edge (5);
    \draw[->] (5) edge (6);
    \draw[->] (6) edge (2);
    \draw[->] (2) edge (10);
    \draw[->] (10) edge (13);
    \draw[->] (13) edge (14);
    \draw[->] (14) edge (11);
    \draw[->] (11) edge (12);
    \draw[->] (12) edge (4);
    \end{scope}
    \begin{scope}[>={Straight Barb[blue]}, every edge/.style={draw, color=blue,very thick}]
    \draw[->] (1) edge (9);
    \draw[->] (9) edge (10);
    \draw[->] (10) edge (11);
    \draw[->] (11) edge (3);
    \draw[->] (3) edge (7);
    \draw[->] (7) edge (8);
    \draw[->] (8) edge (4);

    \end{scope}
\end{tikzpicture}
\caption{Network used as counterexample to the supermodularity of path additions to transport networks (Theorems~\ref{theo:mc_supermod},~\ref{theo:so_ue_supermod}). In this example we have a trip spanning tree (shown in black) and two trip path graphs (shown in blue and orange). Fig.\ref{fig:counterexample_steps} explains this counterexample step by step.}
\label{fig:counterexample_new}
\end{figure}
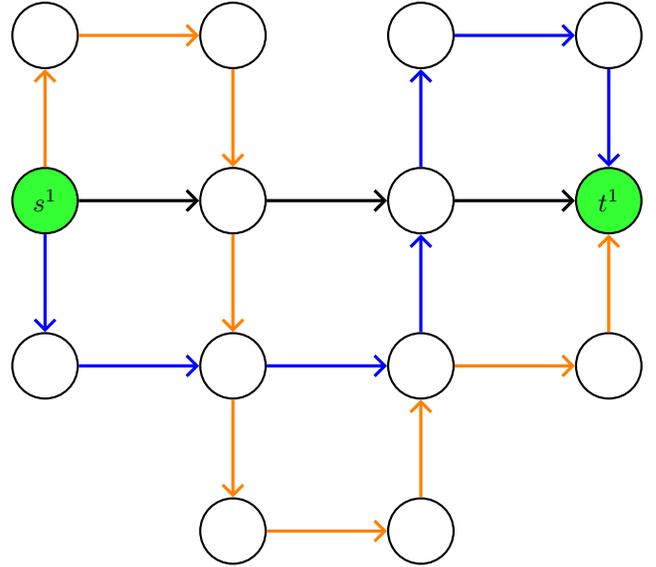

If one can prove that a function $\Lambda : 2^\mathcal{C} \mapsto \R_{>0}$ is not supermodular on a set $\mathcal{C}' \subseteq \mathcal{C}$, it follows that $\Lambda$ is not supermodular on $\mathcal{C}$. Therefore, to increase the generality of our proofs, we restrict the set of possible path additions by considering only trip path graphs that do not share any edges among them and with the trip spanning tree and have the same cost function and parameters for each edge. The sets $\mathcal{G}$ and $\mathcal{H}$ are thus restricted to
\begin{equation*}
\begin{split}
\mathcal{G}' =  \: & \{  G_x | G_x \in \mathcal{G} \: \wedge \\
& \forall G_y \in \mathcal{G}' \:
( G_y \neq G_x \rightarrow  \mathcal{E}_x \cap \mathcal{E}_y \cap \mathcal{E}_\mathcal{I} = \varnothing \: \wedge \\
& \forall (i,j) \in \mathcal{E}_x, \forall (k,z) \in \mathcal{E}_y \: ( u_{ij} = u_{kz} \wedge c_{ij} = c_{kz}  )  )  \}
\end{split}
\end{equation*}
\begin{equation*}
\begin{split}
\mathcal{H}' & =  \{G_x | G_x \in \mathcal{H} \: \wedge \\
& \forall G_y \in \mathcal{H}' \:
( G_y \neq G_x \rightarrow  \mathcal{E}_x \cap \mathcal{E}_y \cap \mathcal{E}_\mathcal{I} = \varnothing \: \wedge \\
& \forall (i,j) \in \mathcal{E}_x, \forall (k,z) \in \mathcal{E}_y, \forall t \in \R_{\geqslant 0} \: ( c_{ij}(t) = c_{kz}(t)  )  )  \}.
\end{split}
\end{equation*}

\begin{theorem}
    $\Lambda_{\mathrm{MC}}$ is \textbf{not} a supermodular set function of the set $\mathcal{G}'$.
    \label{theo:mc_supermod}
\end{theorem}
\begin{proof}
We present a proof by counterexample. It is show that this is valid also in the case of only one trip. Thus, it is sufficient to show that we can find $\{x\}, \mathcal{Y} \subseteq \mathcal{G}'$ such that
\begin{equation}
    \Lambda_{\mathrm{MC}}(\mathcal{\varnothing}) - \Lambda_{\mathrm{MC}}(\{x\}) < \Lambda_{\mathrm{MC}}(\mathcal{Y}) - \Lambda_{\mathrm{MC}}(\mathcal{Y} \cup \{x\}).
    \label{eq:counterexample}
\end{equation}

Let us consider the graph depicted in Fig.~\ref{fig:counterexample_new}. This transport network is composed by: the trip spanning tree $G_\mathcal{I}$, shown in black, the trip path graph $ \{G_{1_y}\} = \mathcal{Y}$, shown in blue, and the trip path graph $G_{1_x} = x$, shown in orange.

The network has the following parameters:
\begin{itemize}
    \item There is only one trip $\mathcal{M} = \{(s^1,t^1,1)\}$
    \item Every edge has a capacity $\geqslant 1$
    \item Edges in the trip spanning tree have a cost of 3
    \item All other edges have a cost of 1.
\end{itemize}
$\Lambda_{\mathrm{MC}}(\mathcal{\varnothing}) = 9$ since all the agents are forced on the only available path (provided by $G_\mathcal{I}$). $\Lambda_{\mathrm{MC}}(\{x\}) = 9$ since the addition provides three new paths, all with the same cost of 9. $\Lambda_{\mathrm{MC}}(\mathcal{Y}) = 7$ since now the blue path has the best cost and it can sustain all the flow. $\Lambda_{\mathrm{MC}}(\mathcal{Y} \cup \{x\}) = 5$ since the blue and orange paths compose to form a solution with cost 5.

Therefore, Equation~\eqref{eq:counterexample} holds. The counterexample is furtherly illustrated in Fig.~\ref{fig:counterexample_steps}.
\end{proof}
\begin{figure}[ht!]
    \centering
    \includegraphics[width=\columnwidth, trim={1mm 1mm 1mm 1mm},clip]{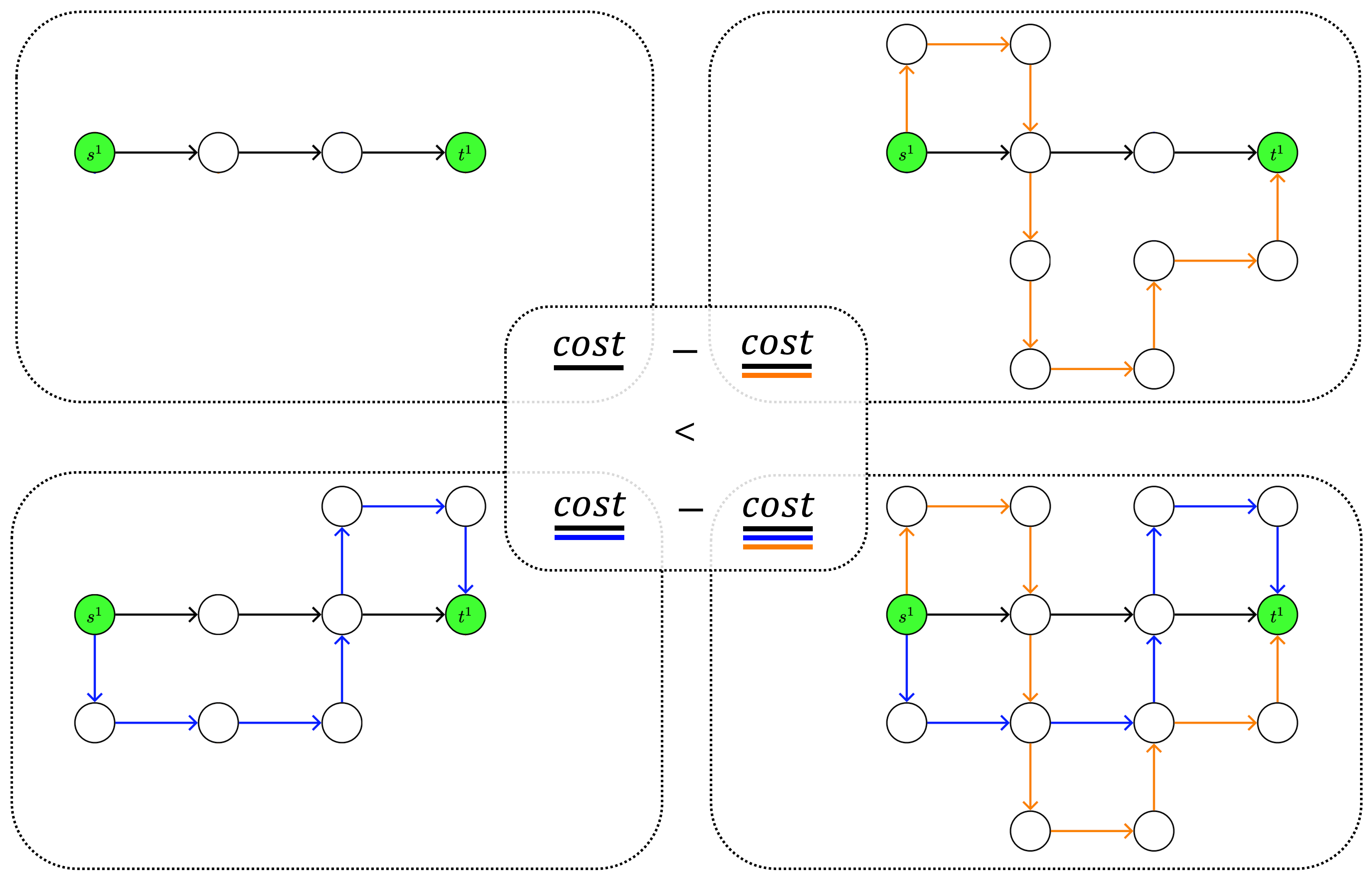}
    \caption{ Visualisation of our counterexample (Theorems~\ref{theo:mc_supermod},~\ref{theo:so_ue_supermod}) illustrating the non-supermodularity of path additions to transport networks. The cost represents the total travel time. In this scenario, the utility (shown as cost difference)  of adding the orange path to the black-and-blue network \textit{is greater} than the utility of adding the orange path just to the black subnetwork. This shows that the utility is not diminishing and, thus, the total travel time of coordinating agents is not supermodular with respect to path additions.}
    \label{fig:counterexample_steps}
\end{figure}
\begin{theorem}
    $\Lambda_{\mathrm{SO}}$ and $\Lambda_{\mathrm{UE}}$ are \textbf{not} supermodular set functions of the set $\mathcal{H}'$.
    \label{theo:so_ue_supermod}
\end{theorem}
\begin{proof}

The counterexample and the network structure are the same as in the proof of Theorem~\ref{theo:mc_supermod}. The parameters are the following:
\begin{itemize}
    \item There is only one trip $\mathcal{M} = \{(s^1,t^1,5)\}$
    \item Every edge has a capacity $u_{ij} = 10$
    \item Edges have a length $l_{ij} = 1$ and a speed limit $v^{\mathrm{max}}_{ij} = 1$ 
    \item The Greenshields affine velocity function (Sec.~\ref{sec:greenshields}) is used as the edge cost function.
\end{itemize}
To compute the $\Lambda_{\mathrm{SO}}$ and $\Lambda_{\mathrm{UE}}$ values of Equation~\eqref{eq:counterexample}, we use the Frank-Wolfe convex optimisation algorithm\footnote{\url{https://github.com/matteobettini/Traffic-Assignment-Frank-Wolfe-2021}}~\cite{leblanc1975efficient}. We confirm the obtained values by exact hand calculations. When solving Equation~\eqref{eq:counterexample} for $\Lambda_{\mathrm{SO}}$ we now obtain:
\begin{equation*}
\begin{gathered}
     \Lambda_{\mathrm{SO}}(\mathcal{\varnothing}) - \Lambda_{\mathrm{SO}}(\{x\}) < \Lambda_{\mathrm{SO}}(\mathcal{Y}) - \Lambda_{\mathrm{SO}}(\mathcal{Y} \cup \{x\}) \\
    30 - 29.28 < 27.47 - 24.97.
\end{gathered}
\end{equation*}
While for $\Lambda_{\mathrm{UE}}$ we obtain:
\begin{equation*}
\begin{gathered}
\Lambda_{\mathrm{UE}}(\mathcal{\varnothing}) - \Lambda_{\mathrm{UE}}(\{x\}) < \Lambda_{\mathrm{UE}}(\mathcal{Y}) - \Lambda_{\mathrm{UE}}(\mathcal{Y} \cup \{x\}) \\
    30 - 30 < 30 - 26.66.
\end{gathered}
\end{equation*}
Both equations hold.

The same result can be shown with the BPR cost function (Sec.~\ref{sec:bpr}) for various choices of $\alpha_{ij}$ and $\beta_{ij}$.
\end{proof}

While the non-supermodularity of $\Lambda_{\mathrm{UE}}$ was foreseeable given its non-monotonicity, the non-supermodularity of path additions for cooperative routing ($\Lambda_{\mathrm{MC}}$ and $\Lambda_{\mathrm{SO}}$) characterises an important insight for transport network designers.

\subsection{Supermodularity of parallel paths}

Having proven the non-supermodularity of our functions on $\mathcal{G}'$ and $\mathcal{H}'$, we are now concerned with the following question: \textit{is there a subset of the space of possible path additions where we can prove supermodularity?} The answer to this question is positive. We will now illustrate such subsets and the relative supermodularity proofs for MC (Sec.~\ref{sec:MC_parallel}) and for SO and UE (Sec.~\ref{sec:identical_parallel}).

\subsubsection{MC}
\label{sec:MC_parallel}

To prove supermodularity in the MC problem, we now consider the case of parallel trip path graphs. In particular, we also restrict our attention to the single trip case, as, due to the parallel nature of the paths we consider, each trip can be treated independently.

Thus, given a trip $(s,t,d)$, we restrict the space of possible trip path graph additions to 
\begin{multline*}
    \mathcal{G}''  = \{G_x | G_x \in \mathcal{G} \wedge \forall (i,j) \in \mathcal{E}_x (u_{ij} \geqslant d) \: \wedge  \forall G_y \in \mathcal{G}'' \\
    ( G_y \neq G_x \rightarrow  \mathcal{E}_x \cap \mathcal{E}_y \cap \mathcal{E}_\mathcal{I} = \varnothing \: \wedge \mathcal{V}_x \cap \mathcal{V}_y \cap \mathcal{V}_\mathcal{I} = \{s,t\}  ) \},
\end{multline*}
which means that we only consider parallel trip path graphs with different costs which can individually hold all the demand.

The graphs considered in this section are obtained after multiple additions (taken from $\mathcal{G}''$) to a trip spanning tree. They will be represented using two nodes: origin ($s$) and destination ($t$), and a set of directed edges, representing the parallel paths connecting $s$ with $t$. Each path has a cost $c_p = \sum_{(i,j) \in p}{c_{ij}}$ and a capacity $u_p = \min_{(i,j) \in p}{u_{ij}}$. An example of such network can be seen in Figure~\ref{fig:parallel}.

\begin{figure}[ht]
\centering
 \begin{tikzpicture}[node distance={6cm}, thick, 
    node/.style = {draw, circle, minimum size=6mm},
    edge_label/.style = {midway, above, sloped, color=black},
    edge/.style = {->, color=black}
    ] 
    
    \node[node] (1) {$s$}; 
    \node[node] (2) [right of=1] {$t$};
    \draw[edge] (1) to  [out=45,in=135,looseness=0.8] node[edge_label] {$c_\mathcal{I}, u_\mathcal{I}$} (2);
    \begin{scope}[>={Straight Barb[black]}]

    \draw[edge] (1) to  [out=22.5,in=157.5,looseness=0.75] node[edge_label] {$c_1, u_1$} (2);
    \draw[edge] (1) to  [out=0,in=180,looseness=0] node[edge_label] {$c_2, u_2$} node[midway,below=0.2cm,sloped] {$\vdots$} (2) ;
    \draw[edge] (1) to  [out=-45,in=-135,looseness=0.7] node[midway, below, sloped, color=black] {$c_n, u_n$} node[midway,above=0.3cm,sloped] {$\vdots$} (2);
    \end{scope}

\end{tikzpicture}
\caption{An example of a graph in the single trip parallel paths case. Here we can see the trip spanning tree $G_\mathcal{I}$ (which in this scenario is always consisting of just one path with cost $c_\mathcal{I}$ and capacity $u_\mathcal{I}$
) and some trip path graphs numbered from 1 to $n$. The notation on the edges shows cost and capacity of each path, separated by a comma.}
\label{fig:parallel}
\end{figure}
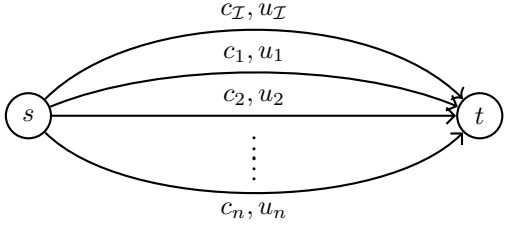

In this scenario, the MC assignment becomes trivial, as the best solution is always to route the whole flow demand on the minimum cost path available in the graph. Hence, we can rewrite $\Lambda_{\mathrm{MC}}$ as follows:
\begin{equation}
    \Lambda_{\mathrm{MC}}(\mathcal{A}) = {\sum_{p \in P_{\mathcal{A} \oplus \mathcal{I}}} x^{\mathrm{MC}}_p c_p} =
    d\min{\left \{c_p |p \in P_{\mathcal{A} \oplus \mathcal{I}} \right\}}.
    \label{eq:parallel_special}
\end{equation}

\begin{lemma}
Given the sets of paths $P_\mathcal{A}$ and $P_\mathcal{B}$ obtained respectively from graphs $G_\mathcal{A}, G_\mathcal{B}$, we have that $P_\mathcal{A}, P_\mathcal{B} \subseteq P_{\mathcal{A} \oplus \mathcal{B}}$.
\label{lemma:sets}
\end{lemma}
\begin{proof}
The proof follows from the fact that graph $G_{\mathcal{A} \oplus \mathcal{B}}$ is defined as $G_{\mathcal{A} \oplus \mathcal{B}} = G_{\mathcal{A}} \cup G_{\mathcal{B}}$ and thus we can write $  P_{\mathcal{A}} \cup P_{ \mathcal{B}} \subseteq P_{\mathcal{A} \oplus \mathcal{B}}$.  Hence, we can rewrite the thesis of Lemma~\ref{lemma:sets} as $P_\mathcal{A}, P_\mathcal{B} \subseteq P_{\mathcal{A}} \cup P_{ \mathcal{B}}$ which is true by the definition of the union over sets. 
\end{proof}

\begin{theorem}
    \label{theo:parallel_agnostic}
     $\Lambda_{\mathrm{MC}}$ is a supermodular set function of the set the set $\mathcal{G}''$.
\end{theorem}

\begin{proof}
The proof can be found in~\ref{sec:proof_parallel_agnostic}.
\end{proof}

The proof depends on the fact that each path can sustain the entire flow. Therefore, the addition of a new path can either:
\begin{itemize}
    \item  Make no change (because it has a higher cost than the best path already in the graph) and will continue to make no change when considering adding it to a superset of the graph.
    \item Make an improvement (because it is the new best path). The improvement is its cost minus the previous best cost. Adding this path to superset of the graph will at best make the same improvement (if the superset graph has the same best cost) or a smaller improvement (if the superset graph has a better best cost than previously).
\end{itemize}

\subsubsection{SO and UE}
\label{sec:identical_parallel}

To prove supermodularity in the SO and UE problem, we consider the case of parallel identical trip path graphs. In particular, we restrict our attention to the single trip case, as, due to the parallel nature of the paths we consider, each trip can be treated independently. The set of possible path additions is $\mathcal{H}'' \subseteq \mathcal{H}$, defined as
\begin{multline*}
    \mathcal{H}'' = \{G_x | G_x \in \mathcal{H} \: \wedge \forall G_y \in \mathcal{H}'' \: ( G_y \neq G_x \rightarrow  \\ 
    \mathcal{E}_x \cap \mathcal{E}_y \cap \mathcal{E}_\mathcal{I} = \varnothing \: \wedge  \mathcal{V}_x \cap \mathcal{V}_y \cap \mathcal{V}_\mathcal{I} = \{s,t\} \: \wedge \\
    \forall p \in P_x, \forall k \in P_y, \forall t \in \R_{\geqslant 0} \:
    ( c_p(t) = c_k(t) ) \}.
\end{multline*}

All the paths being identical, we can treat them as edges with the same length $l$, the same maximum velocity $v_{\mathrm{max}}$ and the same capacity $u$. Thanks to these assumptions, we can use the following flow assignment:
\begin{equation}
    x_p = \frac{d}{n_{P}} \quad \forall p \in P,
    \label{eq:ass_parallel}
\end{equation}
where $n_{P} = |P|$ is the number of parallel paths available. We show that this assignment is optimal for both SO and UE on $\mathcal{H}''$ when using Greenshields' cost function.

\begin{lemma}
Given an instance $(G_{\mathcal{A} \oplus \mathcal{I}}, \mathcal{M}, \mathbf{c})$ where $\mathcal{A} \subseteq \mathcal{H}''$, $\mathcal{M} = \{(s,t,d)\}$, and $\mathbf{c}$ is a vector of Greenshields affine velocity functions, assignment~\eqref{eq:ass_parallel} is an optimal solution for the UE problem.
\label{lemma:ue_parallel}
\end{lemma}

\begin{proof}
By Lemma~\ref{lemma:uer} we know that if we can show that $\forall m \in \mathcal{M}$ and $p_1,p_2 \in P_{\mathcal{A} \oplus \mathcal{I}}^m$ with $x_{p_1} > 0$, we have $c_{p_1}(x_{p_1}) \leqslant c_{p_2}(x_{p_2})$ then the assignment is at UE.

Given that all paths are identical and, under assignment~\eqref{eq:ass_parallel}, all flows are positive, we can rewrite
\begin{equation*}
    c_{p_1}(x_{p_1}) \leqslant c_{p_2}(x_{p_2})
\end{equation*}
as
\begin{equation*}
    \frac{l}{v_{\mathrm{max}}\left (1-\frac{d}{|P_{\mathcal{A} \oplus \mathcal{I}}|u}\right)} \leqslant
    \frac{l}{v_{\mathrm{max}}\left (1-\frac{d}{|P_{\mathcal{A} \oplus \mathcal{I}}|u}\right)},
\end{equation*}
which is always true. 
\end{proof}

\begin{lemma}
Given an instance $(G_{\mathcal{A} \oplus \mathcal{I}}, \mathcal{M}, \mathbf{c})$ where $\mathcal{A} \subseteq \mathcal{H}''$ and $\mathcal{M} = \{(s,t,d)\}$, and $\mathbf{c}$ is a vector of Greenshields affine velocity functions, assignment~\eqref{eq:ass_parallel} is an optimal solution for the SO problem.
\label{lemma:sor_parallel}
\end{lemma}

\begin{proof}
By Lemma~\ref{lemma:sor} we know that if we can show that $\forall m \in \mathcal{M}$ and $p_1,p_2 \in P_{\mathcal{A} \oplus \mathcal{I}}^m$ with $x_{p_1} > 0$, we have $k'_{p_1}(x_{p_1}) \leqslant k'_{p_2}(x_{p_2})$ than the assignment is at UE.

Given that all paths are identical and, under assignment~\eqref{eq:ass_parallel}, all flows are positive, we can rewrite

\begin{equation*}
    k'_{p_1}(x_{p_1}) \leqslant k'_{p_2}(x_{p_2})
\end{equation*}
as
\begin{equation*}
\begin{aligned}
    \left (\frac{d}{|P_{\mathcal{A} \oplus \mathcal{I}}|} \cdot \frac{l}{v_{\mathrm{max}}\left (1-\frac{d}{|P_{\mathcal{A} \oplus \mathcal{I}}|u}\right)} \right )' \leqslant \\
   \left (\frac{d}{|P_{\mathcal{A} \oplus \mathcal{I}}|} \cdot \frac{l}{v_{\mathrm{max}}\left (1-\frac{d}{|P_{\mathcal{A} \oplus \mathcal{I}}|u}\right)} \right )',
  \end{aligned}
\end{equation*}
which is always true. 
\end{proof}

\begin{theorem}
$\Lambda_{\mathrm{SO}}$ and $\Lambda_{\mathrm{UE}}$ are supermodular set functions of the set $\mathcal{H}''$.
\label{theo:parallel_routing}
\end{theorem}

\begin{proof}
The proof can be found in~\ref{sec:proof_parallel_routing}.
\end{proof}

The proof leverages the concept that, after every new addition, the flow will split in equal parts on the available paths. Thus, the benefit of an addition when there are fewer paths available is intuitively greater than the benefit of such addition in a network with more paths.

\section{CONCLUSION}
\label{sec:conclusion}
To conclude, we have analysed and proven several fundamental properties of path additions to transport networks. We have done this for the most popular routing formulations in traffic research. It is shown that, while the total travel time of cooperative agents (MC and SO) is monotonic non-increasing with respect to path additions, it is not supermodular. This impossibility result constitutes a fundamental insight for transport network designers as it informs that approximation algorithms for this problem cannot be designed by leveraging supermodularity in the general case.

On the other hand, we have introduced special cases that consider parallel path additions, where supermodularity holds. These instances could map to real-world scenarios where parallel paths are guaranteed by using dedicated lanes that are either physically or virtually enforced.

Our investigation and categorisation of network design properties under different routing paradigms was motivated by the prospective impact of CAVs on current transport networks. It is conceivable that the widespread adoption of CAVs could lead to a redesign of our mobility infrastructure. More work is needed in this area to understand how network design solutions can be devised and applied to the problem of designing future mobility systems and we believe that the present work constitutes a step towards this overarching goal.


\section*{ACKNOWLEDGMENT}
We gratefully acknowledge the support of ERC Project 949940 (gAIa) and of ARL grant DCIST CRA W911NF-17-2-0181.

\section{APPENDIX}
\subsection{Notation used in the paper}
\label{sec:notation}
\begin{table}[H]
\begin{center}
\resizebox{0.84\columnwidth}{!}{%
\begin{tabular}{c c p{10cm}}
\toprule
$G_x$ & Graph $x$\\
$\mathcal{V}_x$ & Node set of graph $x$\\
$\mathcal{E}_x$ & Edge set of graph $x$\\
$P_x$ & Path set of graph $x$\\
$P_x^m$ & Path set of graph $x$ w.r.t. trip $m$\\
$P_{y|x}^m$ & Paths added to graph $x$ by the addition of graph $y$ w.r.t. trip $m$\\
$m \in \mathcal{M}$ & Trip $m$ in the set of all trips\\
$s^m$ & Source of trip $m$\\
$t^m$ & Destination of trip $m$\\
$d^m$ & Demand of trip $m$\\
$d_i^m$ & Demand of trip $m$ at node $i$\\
$c_{ij}$ & Cost of edge $(i,j)$\\
$c_p$ & Cost of path $p$\\
$u_{ij}$ & Capacity of edge $(i,j)$\\
$u_{p}$ & Capacity of path $p$\\
$l_{ij}$ & Length of edge $(i,j)$\\
$v_{ij}^{\mathrm{max}}$ & Maximum velocity on edge $(i,j)$\\
$x_{ij}$ & Flow on edge $(i,j)$\\
$x_{ij}^m$ & Flow on edge $(i,j)$ w.r.t. trip $m$\\
$x_p$ & Flow on path $p$\\
$\mathbf{x}$ & Vector of all path flows\\
\midrule
$G_{\mathcal{I}}$ & Trip spanning tree\\
$G_{\mathcal{T}}$ & Template graph\\
$G_{m_x}$ & Trip path graph $x$ w.r.t. trip $m$\\
\midrule
$\mathcal{G}$ &  Set of all path additions with constant cost\\
$\mathcal{H}$ & Set of all path additions with flow-dependent cost\\
$\mathcal{G}'$ & Subset of $\mathcal{G}$, where paths do not share edges\\ 
& and all edges have the same cost and capacity\\
$\mathcal{H}'$ & Subset of $\mathcal{H}$, where paths do not share edges\\
& and all edges have the same properties\\
$\mathcal{G}''$ & Subset of $\mathcal{G}'$, where paths do not share nodes\\
& and can sustain the whole flow demand\\
$\mathcal{H}''$ &  Subset of $\mathcal{H}'$, where paths do not share nodes \\ & and have the same properties\\
\midrule
$\Lambda_{\mathrm{MC}}$ & Objective function that takes a graph as input \\
& and computes the total travel time using MC routing\\
$\Lambda_{\mathrm{SO}}$ & Objective function that takes a graph as input \\
& and computes the total travel time using SO routing \\
$\Lambda_{\mathrm{UE}}$ & Objective function that takes a graph as input \\
& and computes the total travel time using UE routing \\
\bottomrule
\end{tabular}}
\end{center}
\label{tab:notation}
\end{table}

\subsection{Proof of Theorem~\ref{theo:so_non_increasing}}
\label{sec:proof_so_non_inreasing}

\setcounter{theorem}{1}
\begin{theorem}
    $\Lambda_{\mathrm{SO}}$ is a monotonic non-increasing set function of the set $\mathcal{H}$.
\end{theorem}

\begin{proof}
We need to show that given two sets of trip path graphs $\mathcal{A}$ and $\mathcal{B}$ such that $\mathcal{A} \subseteq \mathcal{B} \subseteq \mathcal{H}$, we have that $\Lambda(\mathcal{B}) \leqslant \Lambda(\mathcal{A})$.

We can rewrite $\Lambda_{\mathrm{SO}}(\mathcal{A})$ as
  \begin{equation*}
   {\sum_{p \in P_{\mathcal{A} \oplus \mathcal{I}}} x^{\mathrm{SO}}_p c_p(x^{\mathrm{SO}}_p)} =
  \min_{\mathbf{x}}{\left \{ \sum_{p \in P_{\mathcal{A} \oplus \mathcal{I}}} x_p c_p(x_p) \right \}} ,
\end{equation*}
subject to constraints~\eqref{eq:non_zero_sor} and~\eqref{eq:conservation_sor}. $\Lambda(\mathcal{B}) \leqslant \Lambda(\mathcal{A}) $ thus becomes
    \begin{equation*}
     \min_{\mathbf{x}}{\left \{ \sum_{p \in P_{\mathcal{B} \oplus \mathcal{I}}}  x_p c_p(x_p) \right \}} \leqslant \min_{\mathbf{x}}{\left \{ \sum_{p \in P_{\mathcal{A} \oplus \mathcal{I}}}  x_p c_p(x_p) \right \}}.
    \end{equation*}
    If $\mathcal{A} \subseteq \mathcal{B}$ there exists $\mathcal{Y}$ s.t. $\mathcal{B} = \mathcal{A} \cup \mathcal{Y}$. Therefore we can rewrite the above equation as
    \begin{equation*}
     \min_{\mathbf{x}}{\left \{ \sum_{p \in P_{\mathcal{Y} \oplus \mathcal{A} \oplus \mathcal{I}}}  x_p c_p(x_p) \right \}} \leqslant \min_{\mathbf{x}}{\left \{ \sum_{p \in P_{\mathcal{A} \oplus \mathcal{I}}}  x_p c_p(x_p) \right \}}.
    \end{equation*}
    We know that $P_{\mathcal{Y} \oplus \mathcal{A} \oplus \mathcal{I}} $ can be decomposed into its two disjoint components $ P_{ \mathcal{A} \oplus \mathcal{I}}$ and $ P_{\mathcal{Y} | \mathcal{A} \oplus \mathcal{I}}$. Thus we can rewrite
    
    \begin{equation*}
        \begin{aligned}
            \min_{\mathbf{x}}{\left \{ \sum_{p \in P_{\mathcal{A} \oplus \mathcal{I}}}  x_p c_p(x_p) + \sum_{p \in P_{\mathcal{Y}|\mathcal{A} \oplus \mathcal{I}}}  x_p c_p(x_p)\right \}} \leqslant \\
            \min_{\mathbf{x}}{\left \{ \sum_{p \in P_{\mathcal{A} \oplus \mathcal{I}}}  x_p c_p(x_p) \right \}}.
        \end{aligned}
    \end{equation*}
    If we turn our attention to the left side of the inequality and particularly to the second term of the sum, we can set $x_p = 0 ,\; \forall p \in P_{\mathcal{Y}|\mathcal{A} \oplus \mathcal{I}}$, thus obtaining, with this additional constraint
    
    \begin{equation}
        \begin{aligned}
            \min_{\mathbf{x}}{\left \{ \sum_{p \in P_{\mathcal{A} \oplus \mathcal{I}}} x_p c_p(x_p) \right \}} =
            \min_{\mathbf{x}}{\left \{ \sum_{p \in P_{\mathcal{A} \oplus \mathcal{I}}} x_p c_p(x_p) \right \}}.
        \end{aligned}
        \label{eq:artificial_constraint}
    \end{equation}
    This confirms our theorem since the minimum of a function subject to a constraint is greater or equal to the minimum of the function without said constraint.
\end{proof}

\subsection{Proof of Theorem~\ref{theo:mc_non_increasing}}
\label{sec:proof_mc_non_increasing}

\begin{theorem}
    $\Lambda_{\mathrm{MC}}$ is a monotonic non-increasing set function of the set $\mathcal{G}$.
\end{theorem}
\begin{proof}
The proof follows the same procedure as in Theorem~\ref{theo:so_non_increasing}, in the case of constant costs. Property (4) of a trip spanning tree guarantees that when we impose the additional constraint leading to Equation~\eqref{eq:artificial_constraint}, we do not break constraint~\eqref{eq:capacity2} of MC routing and the problem remains feasible.
\end{proof}

\subsection{Proof of Theorem~\ref{theo:parallel_agnostic}}
\label{sec:proof_parallel_agnostic}

\setcounter{theorem}{5}
\begin{theorem}
     $\Lambda_{\mathrm{MC}}$ is a supermodular set function of the set the set $\mathcal{G}''$.
\end{theorem}

\begin{proof}
We need to show that, given two subsets $\mathcal{A}$ and $\mathcal{B}$ such that $\mathcal{A} \subseteq \mathcal{B} \subseteq \mathcal{G}''$ and $x \in \mathcal{G}'' \setminus \mathcal{B}$, it holds that
\begin{equation*}
    \Lambda(\mathcal{A}) - \Lambda(\mathcal{A} \cup \{x\}) \geqslant \Lambda(\mathcal{B}) - \Lambda(\mathcal{B} \cup \{x\}).
\end{equation*}
Knowing that $\mathcal{A} \subseteq \mathcal{B}$ we can define $\mathcal{B} = \mathcal{Y} \cup \mathcal{A}$ with $\mathcal{A} \cap \mathcal{Y} = \varnothing$. This yields
\begin{equation*}
    \Lambda(\mathcal{A}) - \Lambda(\mathcal{A} \cup \{x\}) \geqslant \Lambda(\mathcal{Y} \cup \mathcal{A}) - \Lambda(\mathcal{Y} \cup \mathcal{A} \cup \{x\}),
\end{equation*}
which, according to Equation~\eqref{eq:parallel_special}, we can rewrite as
\begin{equation}
\begin{aligned}
  d\min{\left \{c_p |p \in P_{\mathcal{A} \oplus \mathcal{I}} \right\}} - d\min{\left \{c_p |p \in P_{x \oplus \mathcal{A} \oplus \mathcal{I}} \right\}} \geqslant \\
  d\min{\left \{c_p |p \in P_{\mathcal{Y} \oplus \mathcal{A} \oplus \mathcal{I}} \right\}} - d\min{\left \{c_p |p \in P_{x \oplus \mathcal{Y} \oplus \mathcal{A} \oplus \mathcal{I}} \right\}}.
\end{aligned}
\label{eq:parallel_supermod}
\end{equation}
The term $d$ can be simplified knowing that $d \geqslant 0$. We proceed to consider the following two cases:

\textit{Case 1:} $\min{\left \{c_p |p \in P_{x \oplus \mathcal{A} \oplus \mathcal{I}} \right\}} \geqslant \min{\left \{c_p |p \in P_{\mathcal{Y} \oplus \mathcal{A} \oplus \mathcal{I}} \right\}}$. This implies that $$\min{\left \{c_p |p \in P_{\mathcal{Y} \oplus \mathcal{A} \oplus \mathcal{I}} \right\}} = \min{\left \{c_p |p \in P_{x \oplus \mathcal{Y} \oplus \mathcal{A} \oplus \mathcal{I}} \right\}},$$
Equation~\eqref{eq:parallel_supermod} becomes
\begin{equation}
\begin{aligned}
  \min{\left \{c_p |p \in P_{\mathcal{A} \oplus \mathcal{I}}  \right\}} - \min{\left \{c_p |p \in P_{x \oplus \mathcal{A} \oplus \mathcal{I}}  \right\}} \geqslant 0.
\end{aligned}
\label{eq:parallel_supermod1}
\end{equation}
By Lemma~\ref{lemma:sets} we know that $P_{\mathcal{A} \oplus \mathcal{I}} \subseteq P_{x \oplus \mathcal{A} \oplus \mathcal{I}}$. Recalling that the minimum of a set is always greater or equal than the minimum of one of its supersets, \eqref{eq:parallel_supermod1} holds.

\textit{Case 2:} $\min{\left \{c_p |p \in P_{x \oplus \mathcal{A} \oplus \mathcal{I}} \right\}} < \min{\left \{c_p |p \in P_{\mathcal{Y} \oplus \mathcal{A} \oplus \mathcal{I}} \right\}}$. This implies that $$\min{\left \{c_p |p \in P_{x \oplus \mathcal{A} \oplus \mathcal{I}} \right\}} = \min{\left \{c_p |p \in P_{x \oplus \mathcal{Y} \oplus \mathcal{A} \oplus \mathcal{I}} \right\}},$$ Equation~\eqref{eq:parallel_supermod} becomes
\begin{equation}
\begin{aligned}
  \min{\left \{c_p |p \in P_{\mathcal{A} \oplus \mathcal{I}}  \right\}}  \geqslant 
  \min{\left \{c_p |p \in P_{\mathcal{Y} \oplus \mathcal{A} \oplus \mathcal{I}} \right\}}.
\end{aligned}
\label{eq:parallel_supermod2}
\end{equation}
By Lemma~\ref{lemma:sets} we know that $P_{\mathcal{A} \oplus \mathcal{I}} \subseteq P_{\mathcal{Y} \oplus \mathcal{A} \oplus \mathcal{I}}$. Recalling that the minimum of a set is always greater or equal than the minimum of one of its supersets, \eqref{eq:parallel_supermod2} holds. 
\end{proof}

\subsection{Proof of Theorem~\ref{theo:parallel_routing}}
\label{sec:proof_parallel_routing}

\begin{theorem}
$\Lambda_{\mathrm{SO}}$ and $\Lambda_{\mathrm{UE}}$ are supermodular set functions of the set $\mathcal{H}''$.
\end{theorem}

\begin{proof}
As we have shown with Lemma~\ref{lemma:ue_parallel} and Lemma~\ref{lemma:sor_parallel} that $\Lambda_{\mathrm{UE}}$ and $\Lambda_{\mathrm{SO}}$ are equivalent on $\mathcal{H}''$, we proceed to prove the theorem for $\Lambda_{\mathrm{SO}}$, the result will be also valid for $\Lambda_{\mathrm{UE}}$.
As before, we need to prove that, given two subsets $\mathcal{A}$ and $\mathcal{B}$ such that $\mathcal{A} \subseteq \mathcal{B} \subseteq \mathcal{H}''$ and $x \in \mathcal{H}'' \setminus \mathcal{B}$, it holds that
\begin{equation*}
    \Lambda_{\mathrm{SO}}(\mathcal{A}) - \Lambda_{\mathrm{SO}}(\mathcal{A} \cup \{x\}) \geqslant \Lambda_{\mathrm{SO}}(\mathcal{B}) - \Lambda_{\mathrm{SO}}(\mathcal{B} \cup \{x\}).
\end{equation*}
Knowing that $\mathcal{A} \subseteq \mathcal{B}$ we can define $\mathcal{B} = \mathcal{Y} \cup \mathcal{A}$ with $\mathcal{A} \cap \mathcal{Y} = \varnothing$. This yields
\begin{equation*}
    \Lambda_{\mathrm{SO}}(\mathcal{A}) - \Lambda_{\mathrm{SO}}(\mathcal{A} \cup \{x\}) \geqslant \Lambda_{\mathrm{SO}}(\mathcal{Y} \cup \mathcal{A}) - \Lambda_{\mathrm{SO}}(\mathcal{Y} \cup \mathcal{A} \cup \{x\}).
\end{equation*}
We know that the trip spanning tree $G_\mathcal{I}$ contains only one path $|P_\mathcal{I}| = 1$.
We call the number of paths in $\mathcal{A}$ and $\mathcal{Y}$, $n_a \geqslant 0$ and $n_y \geqslant 0$ respectively. We also know that $|P_{\mathcal{A} \oplus \mathcal{I}}| = n_a + 1$. We fix $n = n_a+1$.

By Lemma~\ref{lemma:sor_parallel} we know that Equation~\eqref{eq:ass_parallel} is an optimal assignment, thus, using the definition of $\Lambda_{\mathrm{SO}}$, we can rewrite
\begin{multline*}
    d c_p \left (\frac{d}{n} \right ) - d c_p \left (\frac{d}{n+1} \right ) \geqslant \\
    d c_p \left (\frac{d}{n + n_y} \right ) - d c_p \left (\frac{d}{n+n_y+1} \right ).
\end{multline*}
By expanding using the Greenshields cost in Equation~\ref{eq:greenshields}, we have
\begin{multline*}
    d\frac{l}{v^{\mathrm{max}}\left (1-\frac{d}{nu}\right)} - d\frac{l}{v^{\mathrm{max}}\left (1-\frac{d}{(n+1)u}\right)} \geqslant \\
    d\frac{l}{v^{\mathrm{max}}\left (1-\frac{d}{(n+n_y)u}\right)} - d\frac{l}{v^{\mathrm{max}}\left (1-\frac{d}{(n+n_y+1)u}\right)}.
\end{multline*}
We know that $d,l,v^{\mathrm{max}},u \geqslant 0 $. Thus, we can apply a series of rewrites:
\begin{multline*}
    \frac{1}{(nu-d)((n+1)u-d)}  \geqslant \\
    \frac{1}{((n+n_y)u-d)((n+n_y+1)u-d)}.
\end{multline*}
Thanks to Property (4) of a trip spanning tree, we know that $u \geqslant d$, thus both denominators are non-negative and we have 
\begin{equation*}
    ((n+n_y)u-d)((n+n_y+1)u-d)  \geqslant 
     (nu-d)((n+1)u-d)
\end{equation*}
\begin{equation*}
    n_yu(n_yu+u+2nu-2d) \geqslant 0,
\end{equation*}
which, by the fact that $n \geqslant 1$ and $u \geqslant d$, is always true. 
\end{proof}

\bibliographystyle{bib/IEEEtran} 
\small{
\bibliography{bib/IEEEabrv,bibliography}}

\end{document}